\documentclass[11pt]{article}

\RequirePackage{amssymb} 
\RequirePackage{amsmath}
\RequirePackage{latexsym}
\RequirePackage{mathrsfs}
\RequirePackage{varioref}

\DeclareMathOperator*{\cov}{cov}

\DeclareMathOperator*{\supp}{supp}

\newcommand{\R}{\ensuremath{\mathbb{R}}}

\newcommand{\Exp}{\ensuremath{\mathbb{E}}} 
\newcommand{\Prob}{\ensuremath{\mathbb{P}}} 

\usepackage{bbm}
\newcommand{\indicator}{\ensuremath{\mathbbm{1}}}

\DeclareMathAlphabet{\mathpzc}{OT1}{pzc}{m}{it}

\newcommand{\normal}{\ensuremath{\mathcal{N}}}

\makeatletter
\newcommand*{\indep}{%
  \mathbin{%
    \mathpalette{\@indep}{}%
  }%
}
\newcommand*{\nindep}{%
  \mathbin{
    \mathpalette{\@indep}{\not}
  }%
}
\newcommand*{\@indep}[2]{%
  \sbox0{$#1\perp\m@th$}
  \sbox2{$#1=$}
  \sbox4{$#1\vcenter{}$}
  \rlap{\copy0}
  \dimen@=\dimexpr\ht2-\ht4-.2pt\relax
  \kern\dimen@
  {#2}%
  \kern\dimen@
  \copy0 
} 
\makeatother

\usepackage{color}

\usepackage{xr-hyper}
\pdfoutput=1
\usepackage{hyperref}
\hypersetup{
    colorlinks=true,
    linkcolor=black,
    citecolor=black,
    filecolor=black,
    urlcolor=black,
}

\usepackage{wrapfig}
\usepackage{tikz-cd}
\usepackage{tikz}
\usepackage{pgfplots}

\RequirePackage{amsthm}

\theoremstyle{definition}

\newtheorem{proposition}{Proposition}
\newtheorem{lemma}{Lemma}

\newtheorem{definition}{Definition}
\newtheorem{theorem}{Theorem}
\newtheorem{corollary}{Corollary}

\newcounter{partialIndepSection}
\newcommand{\aAssump}{A\arabic{partialIndepSection}}
\newcounter{misspecificationSection}
\newcommand{\dAssump}{B\arabic{misspecificationSection}}
\newtheoremstyle{theoremSuppressedNumber}{}{}{}{}{\bfseries}{.}{ }{\thmname{#1}\thmnote{ (\mdseries #3)}}
\theoremstyle{theoremSuppressedNumber}
\newtheorem{partialIndepAssump}{Assumption \aAssump \addtocounter{partialIndepSection}{1}}
\newtheorem{misspecAssump}{Assumption \dAssump \addtocounter{misspecificationSection}{1}}

\usepackage{setspace}
\usepackage[longnamesfirst]{natbib}
\usepackage{ulem}
\usepackage{subfigure}
\usepackage{wrapfig}
\usepackage{float}

\usepackage{epstopdf}

\title{\textbf{Choosing Exogeneity Assumptions \\ in Potential Outcome Models}\footnote{This paper supersedes our previous paper titled ``Interpreting Quantile Independence'' (arXiv:1804.10957), as well as the unpublished sections 3.2, 3.4, and 4.1 of our inactive working paper \cite{MastenPoirier2016}. This paper was presented at the University of Wisconsin--Madison, Northwestern University, the 2017 Triangle Econometrics Conference, the 2018 North American Winter Meetings of the Econometric Society, the 2018 Vanderbilt Conference on Identification in Econometrics, and the 2018 ``Inference in Nonstandard Problems'' CEME Conference. We thank audiences at those seminars and conferences, as well as Federico Bugni, Ivan Canay, Andrew Chesher, Joachim Freyberger, Bo Honor\'e, Joel Horowitz, Shakeeb Khan, Ivana Komunjer, Roger Koenker, Chuck Manski, Francesca Molinari, Jim Powell, and Adam Rosen, for helpful conversations and comments. We thank Hongchang Guo for excellent research assistance. Masten thanks the National Science Foundation for research support under Grant No.\ 1943138.}}

\author{Matthew A. Masten\footnote{Department of Economics, Duke University,
        \texttt{matt.masten@duke.edu}} \qquad Alexandre Poirier\thanks{
    Department of Economics, Georgetown University,
    \texttt{alexandre.poirier@georgetown.edu}}
}
\date{May 4, 2022}

\begin{document}
\maketitle
\begin{abstract}
There are many kinds of exogeneity assumptions. How should researchers choose among them? When exogeneity is imposed on an unobservable like a potential outcome, we argue that the form of exogeneity should be chosen based on the kind of selection on unobservables it allows. Consequently, researchers can assess the plausibility of any exogeneity assumption by studying the distributions of treatment given the unobservables that are consistent with that assumption. We use this approach to study two common exogeneity assumptions: quantile and mean independence. We show that both assumptions require a kind of non-monotonic relationship between treatment and the potential outcomes. We discuss how to assess the plausibility of this kind of treatment selection. We also show how to define a new and weaker version of quantile independence that allows for monotonic treatment selection. We then show the implications of the choice of exogeneity assumption for identification. We apply these results in an empirical illustration of the effect of child soldiering on wages.
\end{abstract}



\onehalfspacing
\normalsize

\newpage
\section{Introduction}\label{sec:intro}

Exogeneity is a critical assumption in much structural or causal empirical work. In general, exogeneity refers to assumptions on the statistical dependence between an observable term and an unobservable term.\footnote{Throughout this paper we use `exogeneity' in the same sense as the treatment effects literature; for example, see \cite{Imbens2004}. This is related to but distinct from the Cowles Commission definition of an exogenous variable as a variable that is determined outside of the model under consideration. See the discussion in \citet[pages 74--75]{HendryMorgan1995}, \citet[page 93]{Imbens1997}, and \citet[footnote 11]{Heckman2000}.} There are many such assumptions, however, including zero correlation, median independence, and full statistical independence.  The choice of a formal definition of exogeneity is not innocuous. Different definitions have different substantive interpretations and different implications for identification, rates of convergence, asymptotic distributions, efficiency bounds, and overidentification.

In the context of potential outcome models, there has been debate over the appropriate choice of exogeneity assumption. \cite{HeckmanIchimuraTodd1998} assume potential outcomes are mean independent of treatment, conditional on covariates. They justify this focus on mean independence by arguing that ``conditional independence assumptions...are far stronger than the mean-independence conditions typically invoked by economists'' (page 262). \citet[page 8]{Imbens2004} agrees that ``this [mean independence] assumption is unquestionably weaker [than full independence]'', but argues that ``in practice it is rare that a convincing case is made for the weaker [mean independence] assumption 2.3 without the case being equally strong for the stronger [full independence assumption].'' The justification he provides is that ``the weaker assumption is intrinsically tied to functional-form assumptions, and as a result one cannot identify average effects on transformations of the original outcome (such as logarithms) without the stronger assumption.'' 

In this paper, we contribute to this debate as follows: We recommend that researchers focus directly on the substantive economic interpretation of the assumption, and the plausibility of its restrictions, rather than assess assumptions based on what they can be used to identify (as in the functional form dependency critique of mean independence) or on mathematical orderings of what implies what (as in the observation that statistical independence implies mean independence, but not vice versa). Specifically, we focus on two of the most common forms of exogeneity assumptions used, quantile independence and mean independence. We provide several results to help researchers assess the plausibility of these exogeneity assumptions. First, in section \ref{sec:whyRelaxIndependence}, we provide a brief informal discussion of the motivations for making exogeneity assumptions. There we note that exogeneity assumptions are typically made on structural unobservables, variables that satisfy some kind of policy or treatment invariance property, like potential outcomes or unobserved ability. In this case, the form of exogeneity depends on the form of treatment selection. Consequently, the plausibility of any exogeneity assumption can be assessed by examining the distributions of treatment given the unobservables that are consistent with that assumption.

Next, in section \ref{sec:characterizationSection}, we characterize these distributions of treatment given the unobservables that are consistent with either quantile independence or mean independence. This characterization shows that both quantile independence and mean independence require specific kinds of non-monotonic treatment selection. In section \ref{sec:weakeningQI} we show how to modify the quantile independence assumption to create a new, weaker exogeneity assumption which allows for more plausible forms of treatment selection, including monotonic treatment selection. We call this assumption \emph{$\mathcal{U}$-independence}. In section \ref{sec:treatmentEffectBounds} we derive identified sets for the average effect of treatment for the treated (ATT) and the quantile treatment effect for the treated (QTT) parameters under both quantile independence and $\mathcal{U}$-independence. These identified sets have a simple, closed form characterization, which makes them easy to use in practice. By comparing these identified sets we show that the identifying power of quantile independence comes from the fact that it only allows for a restrictive kind of non-monotonic selection on unobservables. In section \ref{sec:empirical} we examine these differences in an empirical illustration of the effects of child soldiering on wages based on unconfoundedness. We show that the baseline results are generally robust under quantile independence relaxations of unconfoundedness, but not under $\mathcal{U}$-independence relaxations. This difference highlights both the practical importance of choosing exogeneity assumptions and how our approach can help researchers make this choice. Finally, in section \ref{sec:LatentSelectionModels} we use a Roy model to further illustrate how researchers can assess the plausibility of non-monotonic selection on unobservables.

\section{Choosing Exogeneity Assumptions}\label{sec:whyRelaxIndependence}

There are many different kinds of exogeneity assumptions available to researchers. \cite{Manski1988} and \cite{Powell1994} catalog some of the most common forms, including zero correlation, mean independence, quantile independence, conditional symmetry, statistical independence, and a variety of index conditions. Other kinds of exogeneity assumptions have since been defined, including mean monotonicity (e.g., \citealt{ManskiPepper2000,ManskiPepper2009}, chapter 2 of \citealt{Manski2003}), approximate mean independence (\citealt{Manski2003}, section 9.4), stochastic dominance assumptions (e.g., \citealt{BlundellEtAl2007}), and quantile uncorrelation (\citealt{KomarovaSeveriniTamer2012}), among many others. How should researchers choose among these many options? In this section we discuss one approach to answering this question.

The answer depends on whether the exogeneity assumption is made on a structural unobservable or a reduced form unobservable. This distinction goes back to the earliest work on simultaneous equation models in econometrics (see \citealt{Hausman1983} for a survey) but has been used in recent work as well (e.g. \citealt{BlundellMatzkin2014}). Structural unobservables are variables that satisfy some kind of policy or treatment invariance property, like potential outcomes, unobserved ability, or preferences. Reduced form unobservables are functions of the structural unobservables, and possibly other variables in the model, like realized treatment. Since quantile independence is often imposed on the relationship between treatment variables and structural unobservables, we focus on that case. We briefly discuss exogeneity assumptions for reduced form unobservables in appendix \ref{sec:twoKindsOfUnobs}, along with some additional background and examples.

Let $X$ denote the observed, realized treatment, and let $Y_x$ be a potential outcome where $x$ is a logically possible value of treatment.  Since, by definition, potential outcomes have a meaning and interpretation that does not depend on the realized treatment $X$, any stochastic dependence between them must reflect some form of treatment selection on $Y_x$. That is, it must reflect some form of selection on unobservables, which is described by the distribution of $X \mid Y_x$. Many exogeneity assumptions, such as quantile independence and mean independence, are defined as constraints on the distribution of $Y_x \mid X$. Consequently, to assess the plausibility of these assumptions, we recommend examining the set of distributions of $X \mid Y_x$ that are consistent with the given constraint on the distribution of $Y_x \mid X$.  This allows researchers to use the large literature on treatment selection to assess the plausibility of various exogeneity assumptions. Note that this discussion and recommendation extend to more general structural or causal models where one is considering exogeneity assumptions between a structural unobservable $U$ and an observed variable $X$; $U = Y_x$ is the specific case we focus on here.

\subsubsection*{Descriptive Analysis and Causal Models}\label{sec:descriptiveAnalysis}

Before proceeding, it is important to emphasize the scope of the process we just described: We are interested in assumptions about the dependence structure between observable and unobservable variables in causal models. This does \emph{not} include research whose end goal is a description of the joint distribution of observed random variables, and which does not aim to make causal statements. Such descriptive research studies the relationship between observed variables. For example, suppose $Y$ is an observed outcome and $X$ is an observed covariate. We might define $E$ to be the residual from a linear projection of $Y$ onto $(1,X)$. We can then ask about the statistical relationship between $E$ and $X$: It satisfies zero correlation by construction, but not necessarily other restrictions like mean independence or statistical independence. In this case, however, the joint distribution of $(E,X)$ is always point identified. Hence, at the population level, the precise relationship between these variables is always known. Consequently, there is no need to make or choose assumptions about the stochastic relationship between $E$ and $X$. In contrast, in causal models, exogeneity assumptions typically have substantial identifying power for causal effects or structural parameters.

\section{Characterizing Exogeneity Assumptions}\label{sec:characterizationSection}

In this section, we present our main characterization results. We provide results for two of the most common exogeneity assumptions: quantile independence and mean independence. We focus on binary treatments throughout the paper; we generalize our results to multi-valued discrete and continuous treatments in appendix \ref{sec:multivalTreat}. All of our results also hold if one conditions on an additional vector of observed covariates, as is typically the case in empirical applications, but we omit these for simplicity.

\subsection{The Potential Outcomes Model}

Let $X \in \{0,1\}$ be a binary treatment variable. Let $(Y_1,Y_0)$ denote unobserved potential outcomes. We observe the scalar outcome variable
\begin{equation}\label{eq:potential outcomes}
	Y = X Y_1 + (1-X) Y_0.
\end{equation}
Let $p_x = \Prob(X=x)$ for $x \in \{0,1\}$. We impose the following assumption on the joint distribution of $(Y_1,Y_0,X)$.

\begin{partialIndepAssump}\label{assn:continuity}
For each $x,x' \in \{0,1\}$:
\begin{enumerate}
\item \label{A1_1} $Y_x \mid X=x'$ has a strictly increasing and  continuous distribution function on its support, $\supp(Y_x \mid X=x')$.

\item \label{A1_3} $\supp(Y_x \mid X=x') = \supp(Y_x) = [\underline{y}_x,\overline{y}_x]$ where $-\infty \leq \underline{y}_x < \overline{y}_x \leq \infty$.

\item \label{A1_4} $p_x > 0$.
\end{enumerate}
\end{partialIndepAssump}

Via A\ref{assn:continuity}.\ref{A1_1}, we restrict attention to continuously distributed potential outcomes. A\ref{assn:continuity}.\ref{A1_3} states that the unconditional and conditional supports of $Y_x$ are equal, and are a possibly infinite closed interval. We maintain A\ref{assn:continuity}.\ref{A1_3} for simplicity, but it can be relaxed using similar derivations as in \cite{MastenPoirier2016}. A\ref{assn:continuity}.\ref{A1_4} is an overlap assumption.

In potential outcome models, statistical independence between potential outcomes and treatment is sometimes assumed:
\begin{equation}\label{eq:IndepDefinition}
	Y_x \indep X.
\end{equation}
This assumption, equivalent to random assignment of treatment, can be made for $x =0$, $x=1$, or both. When this assumption is made conditional on covariates, it is often called \emph{unconfoundedness}. As discussed in section \ref{sec:whyRelaxIndependence}, however, there are many other kinds of exogeneity assumptions available in the literature, which are all weaker than full statistical independence. The purpose of this paper is to help researchers choose between these different kinds of exogeneity assumptions. To that end, in the next two subsections we provide characterization results that can help researchers assess the plausibility of these exogeneity assumptions. We then study the identifying power of these different assumptions in section \ref{sec:treatmentEffects}.

\subsection{A Class of Quantile Independence Assumptions}

Quantile independence of the potential outcome $Y_x$ from $X$ at quantile $\tau$ holds if
\begin{equation}\label{eq:quantileBasedDefinition}
	Q_{Y_x \mid X}(\tau \mid 0) = Q_{Y_x \mid X}(\tau \mid 1).
\end{equation}
This assumption is often imposed at a single quantile. For example, imposing \eqref{eq:quantileBasedDefinition} at $\tau = 0.5$ yields median independence. If this holds for all $\tau \in (0,1)$, then $Y_x$ and $X$ are statistically independent. Therefore quantile independence is a relaxation of independence.

It is often more natural to work with cdfs, an inverse of the quantile function.\footnote{For example, see assumption QI on page 731 of \cite{Manski1988} or equation (1.7) on page 2452 of \cite{Powell1994}. Definitions using cdfs and those using quantiles directly (e.g., via equation \eqref{eq:quantileBasedDefinition}) are often equivalent. Throughout this paper we use ``quantile independence'' to mean the cdf-based definition, as is common in the literature.} Say $Y_x$ is $\tau$-cdf independent of $X$ if
\begin{equation}\label{cdfIndependence}
	F_{Y_x \mid X}(\tau \mid 0) = F_{Y_x \mid X}(\tau \mid 1).
\end{equation}
Note that $Y_x$ is quantile independent of $X$ at quantile $\tau$ if and only if $Y_x$ is $Q_{Y_x}(\tau)$-cdf independent of $X$ for continuously distributed $Y_x$. This motivates the following definition.\footnote{See \cite{BelloniChenChernozhukov2017} and \cite{ZhuZhangXu2017} for similar generalizations of quantile independence.}

\begin{definition}
Let $\mathcal{T}$ be a subset of $\R$. Say $Y_x$ is \emph{$\mathcal{T}$-independent} of $X$ if  the cdf independence condition \eqref{cdfIndependence} holds for all $\tau \in \mathcal{T}$.
\end{definition}

With binary treatments, the dependence structure between $X$ and the potential outcome $Y_x$ is fully characterized by the function
\[
	p(y_x) = \Prob(X=1 \mid Y_x=y_x),
\]
which we call the latent propensity score. Full statistical independence of $Y_x$ and $X$, or $\mathcal{T}$-independence with $\mathcal{T} = [\underline{y}_x,\overline{y}_x]$, is equivalent to this latent propensity score being constant:
\[
	p(y_x) = \Prob(X=1)
\]
for almost all $y_x \in \supp(Y_x)$. Analogously, $\mathcal{T}$-independence for $\mathcal{T} \subsetneq [\underline{y}_x,\overline{y}_x]$ is weaker than full independence, and thus partially restricts the shape of $p(y_x)$. The following theorem characterizes the set of latent propensity scores consistent with $\mathcal{T}$-independence.

\begin{theorem}[Average value characterization]\label{thm:AvgValueCharacterization}
Suppose $X$ is binary and A\ref{assn:continuity}.1 holds. Then $Y_x$ is $\mathcal{T}$-independent of $X$ if and only if
\begin{equation}\label{eq:averageValueCondition_main}
	\Exp \big( p(Y_x) \mid Y_x \in (t_1,t_2) \big) = \Prob(X=1)
\end{equation}
for all $t_1, t_2 \in \mathcal{T} \cup \{ \underline{y}_x,\overline{y}_x \}$ with $t_1 < t_2$.
\end{theorem}

The proof, along with all others, is in appendix \ref{sec:proofs}. Theorem \ref{thm:AvgValueCharacterization} says that $\mathcal{T}$-independence holds if and only if for every interval with endpoints in $\mathcal{T} \cup \{ \underline{y}_x,\overline{y}_x \}$ the average latent propensity score over $Y_x \in (t_1,t_2)$ equals the overall average of the latent propensity score, which is $\Prob(X=1)$. Also note that $\Exp(p(Y_x) \mid Y_x \in (t_1,t_2)) = \Prob(X=1 \mid Y_x \in (t_1,t_2))$.

To illustrate theorem \ref{thm:AvgValueCharacterization}, suppose $\mathcal{T} = \{ 0.5 \}$ and $\Prob(X=1) = 0.5$. Further suppose that $Y_x$ is uniformly distributed on $[0,1]$ to simplify the figures. Here we have just a single nontrivial cdf independence condition: median independence. Figure \ref{propensityScoresFig1} plots three different latent propensity scores which are consistent with $\mathcal{T}$-independence under this choice of $\mathcal{T}$; that is, which are consistent with median independence. This figure illustrates several features of such latent propensity scores: The value of $p(y_x)$ may vary over the entire range $[0,1]$. $p$ does not need to be symmetric about $y_x = 0.5$, nor does it need to be continuous. It does need to satisfy equation \eqref{eq:averageValueCondition_main} over the intervals $(t_1,t_2) = (0,0.5)$ and $(t_1,t_2) = (0.5,1)$. Finally, as suggested by the pictures, $p$ must actually be nonmonotonic; we show this in corollary \ref{corr:monotonicPropensityScores} next.

\begin{figure}[t]
\centering
\includegraphics[width=52mm]{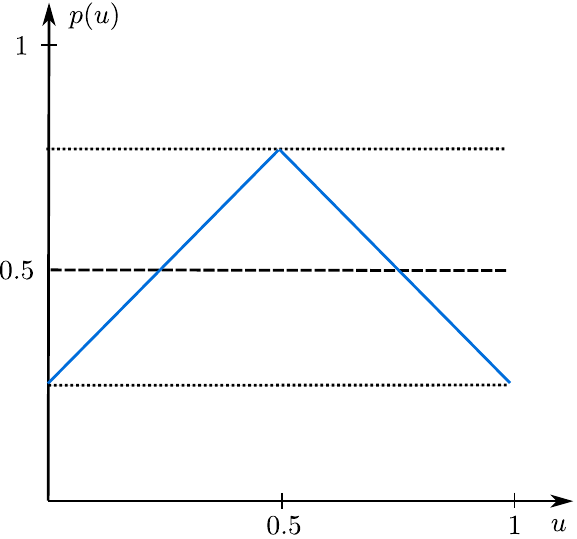}
\hspace{0.2mm}
\includegraphics[width=52mm]{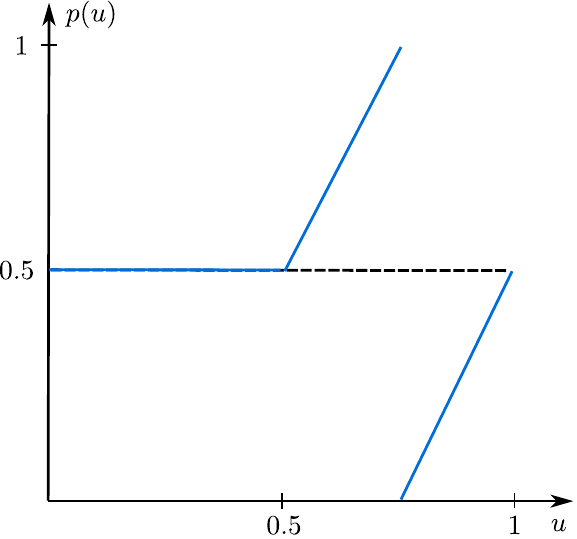}
\hspace{0.2mm}
\includegraphics[width=52mm]{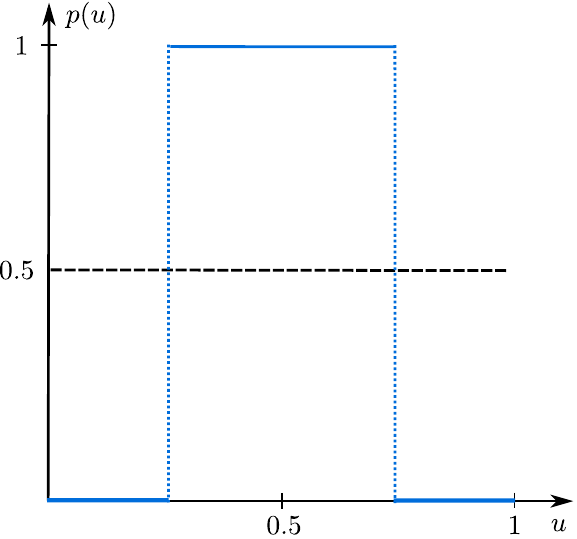}
\caption{Various latent propensity scores consistent with $\mathcal{T} = \{ 0.5 \}$-independence, when $\Prob(X=1) = 0.5$.}
\label{propensityScoresFig1}
\end{figure}

\begin{corollary}\label{corr:monotonicPropensityScores}
Suppose $X$ is binary and A\ref{assn:continuity}.1 holds. Suppose the latent propensity score $p$ is weakly monotonic and not constant on $(\underline{y}_x,\overline{y}_x)$. Then, for all $\tau \in (\underline{y}_x,\overline{y}_x)$, $Y_x$ is not $\tau$-cdf independent of $X$.
\end{corollary}

Corollary \ref{corr:monotonicPropensityScores} shows that any quantile independence assumption rules out all types of monotonic selection, except for the trivially monotonic constant $p(y_x) = \Prob(X=1)$ implied by full statistical independence.

Next we show that imposing multiple quantile independence conditions imposes further non-monotonicity. Say that a function $f$ \emph{changes direction at least $K$ times} if there exists a partition of its domain into $K$ intervals such that $f$ is not monotonic on each interval.

\begin{corollary}\label{corr:SignChanges}
Suppose $X$ is binary and A\ref{assn:continuity}.1 holds. Suppose $Y_x$ is $\mathcal{T}$-independent of $X$. Suppose there exists a version of $p$ without removable discontinuities. Partition $(\underline{y}_x,\overline{y}_x)$ by the sets $\mathcal{Y}_1 = (t_0,t_1)$, $\mathcal{Y}_k = [t_{k-1},t_k)$ for $k=2,\ldots,K$ with $t_0 = \underline{y}_x$, $t_K = \overline{y}_x$,  and such that for each $k$ there is a $\tau_k \in \mathcal{T} \cap \mathcal{Y}_k$. Suppose $p$ is not constant over each set $\mathcal{Y}_k$, $k=1,\ldots,K$. Then $p$ changes direction at least $K$ times.
\end{corollary}

This result says that such latent propensity scores must oscillate up and down at least $K$ times (we assume $p$ does not have removable discontinuities to rule out trivial direction changes).  For example, as in figure \ref{propensityScoresFig1}, suppose we continue to have $\Prob(X=1) = 0.5$ and $Y_x \sim \text{Unif}[0,1]$ but we add a few more isolated $\tau$'s to $\mathcal{T}$. Figure \ref{propensityScoresFig2} shows several latent propensity scores consistent with $\mathcal{T}$-independence when $\mathcal{T}$ has several isolated elements. Consider the figure on the left, with $\mathcal{T} = \{ 0.25, 0.5, 0.75 \}$. Partition $(0,1) = (0,0.4) \cup [0.4,0.6) \cup [0.6,1)$. Then $p$ is not monotonic over each partition set, and each partition set contains one element of $\mathcal{T}$: $0.25 \in (0,0.4)$, $0.5 \in [0.4,0.6)$, and $0.75 \in [0.6,1)$. There are $K=3$ partition sets, and hence the corollary says $p$ must change direction at least 3 times. We see this in the figure since there are 3 interior local extrema. A similar analysis holds for the figure on the right. Overall, these triangular and sawtooth latent propensity scores illustrate the oscillation required by corollary \ref{corr:SignChanges}.

\begin{figure}[t]
\centering
\includegraphics[width=52mm]{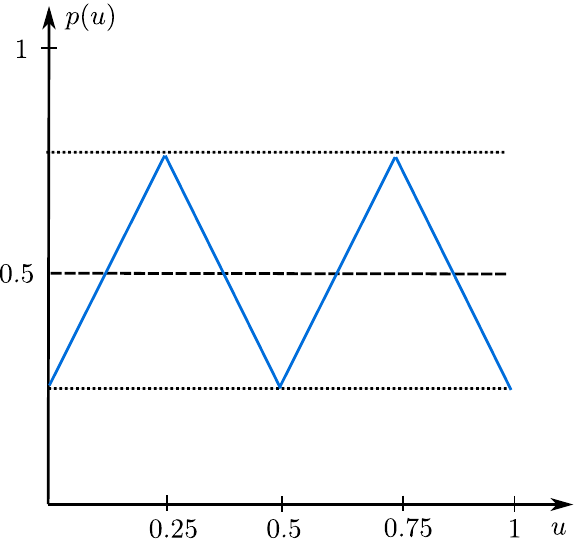}
\hspace{0.2mm}
\includegraphics[width=52mm]{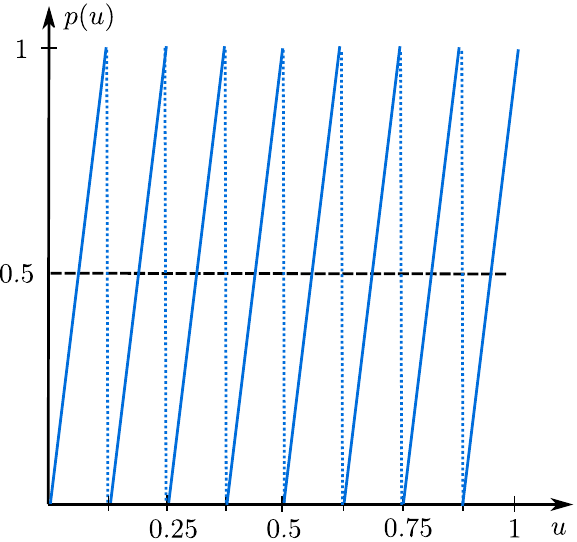}
\caption{Some latent propensity scores consistent with $\mathcal{T}$-independence when $\Prob(X=1) = 0.5$. Left: $\mathcal{T} = \{ 0.25, 0.5, 0.75 \}$. Right: $\mathcal{T} = \{ 0.125, 0.25, 0.375, 0.5, 0.625, 0.75, 0.875 \}$.}
\label{propensityScoresFig2}
\end{figure}

We document one more feature: As long as there is some interval that is not in $\mathcal{T}$ then there is a latent propensity score that takes the most extreme values possible, 0 and 1.

\begin{corollary}\label{corr:TauIndep_ExtremeValues}
Suppose $X$ is binary and A\ref{assn:continuity}.1 and A\ref{assn:continuity}.3 hold. Suppose $[\underline{y}_x,\overline{y}_x] \setminus \mathcal{T}$ contains a non-degenerate interval. Then there exists a latent propensity score which is consistent with $\mathcal{T}$-independence of $Y_x$ from $X$ and for which the sets
\[
	\{y_x \in [\underline{y}_x,\overline{y}_x]: p(y_x) = 0\}
	\qquad \text{and} \qquad
	\{y_x \in [\underline{y}_x,\overline{y}_x]: p(y_x) = 1\}
\]
have positive Lebesgue measure.
\end{corollary}

Consequently, $\mathcal{T}$-independence allows for a kind of extreme imbalance, where there is a positive mass of potential outcome values that only appear in the treatment group and another positive mass of potential outcome values that only appear in the control group.

\subsection{Characterizing Mean Independence}\label{sec:meanIndependence}

Mean independence is another commonly used exogeneity assumption. For example, \cite{HeckmanIchimuraTodd1998} assume potential outcomes are mean independent of treatments, conditional on covariates. As in the previous section, we characterize the constraints this assumption places on the conditional distribution of $X$ given $Y_x$.

\begin{definition}\label{def:meanIndep}
Say $Y_x$ is \emph{mean independent} of $X$ if $\Exp(Y_x \mid X=0) = \Exp(Y_x \mid X=1)$.
\end{definition}

From definition \ref{def:meanIndep} and Bayes' rule, it immediately follows that
\begin{equation}\label{eq:meanIndepIntegral}
	\Exp\left(\frac{Y_x}{\Exp(Y_x)} p(Y_x)\right) = \Prob(X=1),
\end{equation}
assuming $\Exp(Y_x) \neq 0$. Theorem \ref{thm:AvgValueCharacterization} showed that quantile independence constrains the \emph{unweighted} average value of the latent propensity score over certain \emph{subintervals} of its domain. In contrast, equation \eqref{eq:meanIndepIntegral} shows that mean independence constrains a \emph{weighted} average value of the latent propensity score over its entire domain. Equation \eqref{eq:meanIndepIntegral} can be extended to multi-valued and continuous $X$ as in our analysis of quantile independence in section \ref{sec:multivalTreat}; we omit this extension for brevity.

Although mean independence imposes a different constraint on the latent propensity score than quantile independence, it also requires non-constant latent propensity scores to be non-monotonic.

\begin{proposition}\label{prop:mean-indep prop scores}
Suppose $X$ is binary and A\ref{assn:continuity}.1 holds. Suppose $\Exp(|Y_x|) < \infty$. Suppose the latent propensity score $p$ is weakly monotonic and not constant on the interior of its domain. Then $Y_x$ is not mean independent of $X$.
\end{proposition}

Therefore, assuming mean independence rules out all types of monotonic selection, except for the trivially monotonic constant $p(y_x) = \Prob(X=1)$ implied by full statistical independence. 

\subsection{Discussion}

To place these results in context, consider the case where $X$ is an indicator for completing college and $Y_0$ denotes a person's earnings if they do not complete college. $X$ is often thought to be endogenous due to its relationship with ability, which is captured by $Y_0$. In this example, $p(y_0)$ is the proportion of people who complete college, among those with a fixed level of non-graduate earnings. Corollary \ref{corr:monotonicPropensityScores} states that any quantile independence condition rules out nonconstant, weakly monotonic treatment selection. Similarly, proposition \ref{prop:mean-indep prop scores} implies that mean-independence rules out this monotonic treatment selection. They would thus rule out that the proportion who attend college is weakly increasing in the level of non-graduate earnings, unless we assume that college attendance and non-graduate earnings are statistically independent, in which case $p(y_0)$ is constant.

Corollary \ref{corr:SignChanges} would require the probability of attending college to oscillate (or be constant) to accommodate multiple quantile independence conditions. For example, if $K$ quantile independence conditions hold, there must exist $K$ non-graduate earnings thresholds where the effect of non-graduate earnings on college attendance changes sign. Alternatively, college attendance can again be statistically independent of non-graduate earnings. We discuss the plausibility of these oscillations in the context of a Roy Model in section \ref{sec:LatentSelectionModels}.

Corollary \ref{corr:TauIndep_ExtremeValues} characterizes another feature of the latent propensity scores allowed by quantile independence. In our returns to schooling example, a finite number of quantile independence restrictions allow for a strictly positive proportion of people with non-graduate earnings levels for which nobody attends college ($p(y_0) = 0$), and for which everyone attends college ($p(y_0) = 1$). Depending on the context, the existence of these two groups may or may not appear plausible. If it appears implausible, using an assumption that allows for their existence is inefficient compared to an assumption that rules out their existence. Since ruling out their existence is a stronger assumption, it would result in weakly narrower identified sets compared to the overly conservative set obtained under quantile independence alone. To rule out their existence, one could add a constraint such as $p(y_0) \in [c_1, c_2]$ for prespecified $0 < c_1 \leq c_2 < 1$ and derive the identified set for the relevant parameter under quantile independence and this added constraint. Throughout this paper, we emphasize this approach of tailoring exogeneity conditions to the likely and unlikely features of the treatment selection function $p(y_0)$ in the given application. Precisely finding the assumption that best captures the nonexistence of those groups is beyond the scope of this paper, however, since it is application dependent. 

\section{The Identifying Power of Different Exogeneity Assumptions}
\label{sec:treatmentEffects}

In this section we study the implications of the choice of exogeneity assumption for identification. Our first result in section \ref{sec:characterizationSection} shows that quantile independence imposes a constraint on the average value of a latent propensity score. We use this characterization to motivate an assumption weaker than quantile independence, which we call \emph{$\mathcal{U}$-independence}. The difference between these two assumptions is that quantile independence imposes some additional average value constraints on $p(y_x)$ that $\mathcal{U}$-independence does not. In particular, $\mathcal{U}$-independence allows for monotonic treatment selection. Hence the difference between identified sets obtained under these two assumptions is a measure of the identifying power of these additional average value constraints, which are the features of quantile independence that require the latent propensity score to be non-monotonic. Unlike quantile independence, it is not clear to us how to naturally weaken mean independence to allow for monotonic latent propensity scores while still retaining an interpretable assumption with identifying power. For that reason, in this section we only study identification under quantile independence and its weaker version, $\mathcal{U}$-independence.

We focus on two parameters: The average treatment effect for the treated,
\begin{align*}
	\text{ATT}
	&= \Exp(Y_1 - Y_0 \mid X=1) \\
	&= \Exp(Y \mid X=1) - \Exp(Y_0 \mid X=1),
\end{align*}
and the quantile treatment effect for the treated,
\begin{align*}
	\text{QTT}(q)
	&= Q_{Y_1 \mid X}(q \mid 1) - Q_{Y_0 \mid X}(q \mid 1) \\
	&= Q_{Y \mid X}(q \mid 1) - Q_{Y_0 \mid X}(q \mid 1),
\end{align*}
for $q \in (0,1)$. To analyze treatment on the treated parameters, we only need to make assumptions on the relationship between $Y_0$ and $X$. Our analysis can easily be extended to parameters like ATE by imposing $\mathcal{T}$- or $\mathcal{U}$-independence between $Y_1$ and $X$ as well as between $Y_0$ and $X$.

Under statistical independence $Y_0 \indep X$, both the ATT and QTT are point identified. Under $\mathcal{T}$- or $\mathcal{U}$-independence, however, they are generally partially identified. We derive identified sets for ATT and QTT under both classes of exogeneity assumptions in this section. These sets have simple explicit expressions which make them easy to use in practice. In section \ref{sec:empirical} we compare these identified sets in an empirical application. In this application, the estimated identified sets are significantly larger under $\mathcal{U}$-independence, implying that the additional average value constraints inherent in $\mathcal{T}$-independence have substantial identifying power.

\subsection{Weakening Quantile Independence}\label{sec:weakeningQI}

Throughout this section, we focus on the case where $\mathcal{T}$ is an interval. In this case, we show that latent propensity scores consistent with $\mathcal{T}$-independence have two features: (a) they are flat on $\mathcal{T}$ and (b) they are non-monotonic outside the flat regions, such that the average value constraint \eqref{eq:averageValueCondition_main} is satisfied. We use this finding to motivate a weaker assumption which retains feature (a) but drops feature (b). We call this assumption \emph{$\mathcal{U}$-independence}. This new weaker assumption has two uses: First, we can use it as a tool for understanding quantile independence itself. Specifically, by comparing identified sets under quantile independence and under the weaker $\mathcal{U}$-independence we will learn the identifying power of the average value constraints on the latent propensity score. Second, $\mathcal{U}$-independence can be used by itself as a method for relaxing statistical independence and performing sensitivity analysis. We illustrate both of these uses below and in our empirical analysis of section \ref{sec:empirical}.

We begin with the following corollary to theorem \ref{thm:AvgValueCharacterization}.

\begin{corollary}\label{prop:TauImpliesU}
Suppose $X$ is binary and that A\ref{assn:continuity} holds for $Y_0$. Let $\mathcal{T} = [a,b]\subseteq [\underline{y}_0,\overline{y}_0]$. Then $\mathcal{T}$-independence of $Y_0$ from $X$ implies
\begin{equation}\label{eq:UindepEq}
	\Prob(X=1 \mid Y_0=y_0) = \Prob(X=1)
\end{equation}
for almost all $y_0 \in \mathcal{T}$.
\end{corollary}

Corollary \ref{prop:TauImpliesU} shows that $\mathcal{T}$-independence requires the latent propensity score to be constant on $\mathcal{T}$ and equal to the overall unconditional probability of being treated. The first property---that the latent propensity score is flat on $\mathcal{T}$---means that random assignment holds within the subpopulation of units whose untreated outcomes are in the set $\mathcal{T}$; that is, $X \indep Y_0 \mid \{ Y_0 \in \mathcal{T} \}$. Corollary \ref{prop:TauImpliesU} can be generalized to allow $\mathcal{T}$ to be a finite union of intervals, but we omit this for brevity.

This corollary motivates the following definition.

\begin{definition}
Let $\mathcal{U} \subseteq [\underline{y}_x,\overline{y}_x]$ be an interval. Say that $Y_x$ is \emph{$\mathcal{U}$-independent} of $X$ if $\Prob(X=1 \mid Y_x=y_x) = \Prob(X=1)$ for almost all $y_x \in \mathcal{U}$.
\end{definition}

Importantly, unlike $\mathcal{T}$-independence, $\mathcal{U}$-independence allows for monotonic treatment selection. Corollary \ref{prop:TauImpliesU} shows that $\mathcal{T}$-independence implies $\mathcal{U}$-independence with $\mathcal{U} = \mathcal{T}$. The converse does not hold since $\mathcal{T}$-independence requires additional average value constraints to hold, by theorem \ref{thm:AvgValueCharacterization}. In particular, $\mathcal{U}$-independence implies the average value constraint
\[
	\Exp(p(Y_x) \mid Y_x \in (t_1,t_2)) = \Prob(X=1)
	\tag{\ref{eq:averageValueCondition_main}}
\]
for all $t_1,t_2 \in \mathcal{U}$, because it requires that $p(u)$ is constant on $\mathcal{U}$. But $\mathcal{T}$-independence also requires that \eqref{eq:averageValueCondition_main} holds for choices of $t_1$ and $t_2$ in $\mathcal{U} \cup \{\underline{y}_x,\overline{y}_x\}$. That is, $t_1$ and $t_2$ can equal the end points $\underline{y}_x$ or $\overline{y}_x$. Hence it imposes an average value constraint outside of the set $\mathcal{U}$. For example, figure \ref{TauIndepFlatPropScore} shows two latent propensity scores. One satisfies $\mathcal{T}$-independence, but the other only satisfies $\mathcal{U}$-independence. Finally, note that $\mathcal{U}$-independence is a nontrivial assumption only when $\Prob(Y_x \in \mathcal{U}) > 0$. Conversely, $\mathcal{T}$-independence is nontrivial even when $\mathcal{T}$ is a singleton.

\begin{figure}[t]
\centering
\includegraphics[width=52mm]{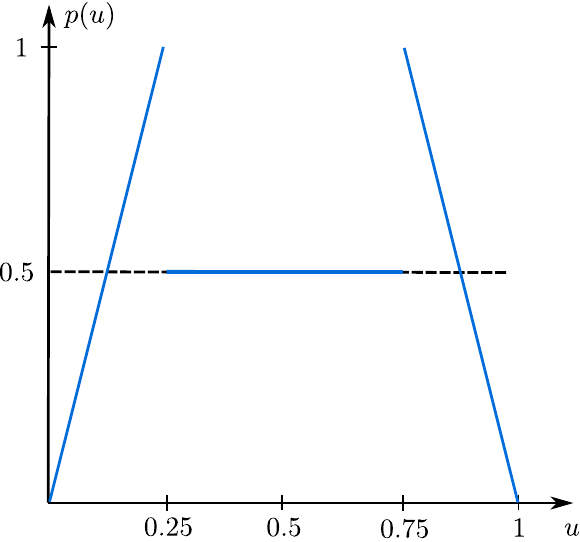}
\hspace{5mm}
\includegraphics[width=52mm]{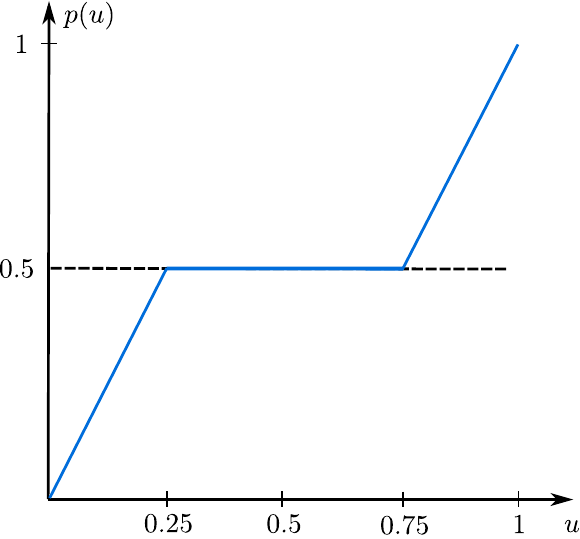}
\caption{Let $\mathcal{T} = \mathcal{U} = [0.25,0.75]$ and $\Prob(X=1) = 0.5$. Also normalize $Y_0$ to be uniform on $(0,1)$ for simplicity. This figure shows a latent propensity score consistent with both $\mathcal{T}$- and $\mathcal{U}$-independence on the left and a latent propensity score consistent with $\mathcal{U}$-, but not $\mathcal{T}$-independence on the right.}
\label{TauIndepFlatPropScore}
\end{figure}

\subsection{The Identified Sets For ATT and $\text{QTT}(q)$}\label{sec:treatmentEffectBounds}

In this subsection we derive sharp bounds on the ATT and $\text{QTT}(q)$ under both $\mathcal{T}$- and $\mathcal{U}$-independence. To do so, it suffices to derive bounds on $Q_{Y_0 \mid X}(q \mid 1)$. 

We show the validity of the following bounds in proposition \ref{prop:quantileTrtBounds} below. Let $\mathcal{T} = \mathcal{U} = [Q_{Y_0}(a),Q_{Y_0}(b)]$ for  $0 < a \leq b < 1$. The $\mathcal{T}$-independence bounds are then defined by
\[
	\overline{Q}_{Y_0 \mid X}^\mathcal{T}(\tau \mid 1)
	=
	\begin{cases}
		Q_{Y \mid X}(a \mid 0)
			&\text{ for $\tau \in (0,a]$} \\
		Q_{Y \mid X}(\tau \mid 0)
			&\text{ for $\tau\in(a,b]$} \\
		Q_{Y \mid X}(1 \mid 0)
			&\text{ for $\tau\in(b,1)$},
	\end{cases}
	\qquad
	\underline{Q}_{Y_0 \mid X}^\mathcal{T}(\tau \mid 1)
	=
	\begin{cases}
		Q_{Y \mid X}(0 \mid 0)
			&\text{ for $\tau\in(0,a]$} \\
		Q_{Y \mid X}(\tau \mid 0)
			&\text{ for $\tau\in(a,b]$} \\
		Q_{Y \mid X}(b \mid 0)
			&\text{ for $\tau \in(b,1)$}.
	\end{cases}
\]
We let $Q_{Y \mid X}(0 \mid x) = \underline{y}_x$ and $Q_{Y \mid X}(1 \mid x) = \overline{y}_x$. For $\mathcal{U}$-independence, there are two cases. First consider the lower bound. If $(1-(b-a))p_1 \leq a$,
\[
	\underline{Q}_{Y_0 \mid X}^\mathcal{U}(\tau \mid 1)
	=
	\begin{cases}
		Q_{Y \mid X}(0 \mid 0)
			&\text{ for $\tau \in (0,1-(b-a)]$} \\
		Q_{Y \mid X}\left( \tau + \dfrac{b-1}{p_0} \mid 0 \right)
			&\text{ for $\tau \in (1-(b-a), 1)$}.
	\end{cases}
\]
If $(1-(b-a))p_1 \geq a$,
\[
	\underline{Q}_{Y_0 \mid X}^\mathcal{U}(\tau \mid 1)
	=
	\begin{cases}
		Q_{Y \mid X}(0 \mid 0)
			&\text{ for $\tau \in \left(0,\dfrac{a}{p_1} \right]$} \\
		Q_{Y \mid X} \left(\tau - \dfrac{a}{p_1} \mid 0 \right)
			&\text{ for $\tau \in \left( \dfrac{a}{p_1},\dfrac{a}{p_1} + b-a \right]$} \\
		Q_{Y \mid X}(b-a \mid 0)
			&\text{ for $\tau \in \left( \dfrac{a}{p_1} + b-a,1 \right)$}.
	\end{cases}
\]
Next consider the upper bound. If $(1-(b-a))p_0 \leq a$,
\[
	\overline{Q}_{Y_0 \mid X}^\mathcal{U}(\tau \mid 1)
	= \begin{cases}
		 Q_{Y \mid X}(1-(b-a) \mid 0)
		 &\text{for $\tau \in \left( 0,1-(b-a) - \dfrac{1-b}{p_1} \right]$} \\
		 Q_{Y \mid X} \left( \tau +\dfrac{1-b}{p_1} \mid 0 \right)
		 &\text{for $\tau \in \left( 1-(b-a) - \dfrac{1-b}{p_1}, 1 - \dfrac{1-b}{p_1} \right]$} \\
		 Q_{Y \mid X}(1 \mid 0)
		 &\text{for $\tau \in \left( 1 - \dfrac{1-b}{p_1}, 1 \right)$.}											\end{cases}
\]
If $(1-(b-a))p_0 \geq a$,
\[
	\overline{Q}_{Y_0 \mid X}^\mathcal{U}(\tau \mid 1)
	=
	\begin{cases}
	Q_{Y \mid X} \left(\tau + \dfrac{a}{p_0} \mid 0 \right)
		&\text{ for $\tau \in (0, b-a]$} \\
	Q_{Y \mid X}(1 \mid 0)
		&\text{ for $\tau \in (b-a,1)$}.
	\end{cases}
\]

\begin{proposition}\label{prop:quantileTrtBounds}
Let A\ref{assn:continuity} hold. Suppose $Y_0$ is $\mathcal{T}$-independent of $X$ with $\mathcal{T} = [Q_{Y_0}(a),Q_{Y_0}(b)]$, $0 < a \leq b < 1$. Suppose the joint distribution of $(Y,X)$ is known. Let $q \in (0,1)$. Then 
\begin{equation}\label{eq:Y0quantileIdentifiedSet}
	Q_{Y_0 \mid X}(q \mid 1) \in \left[ \underline{Q}_{Y_0 \mid X}^\mathcal{T}(q \mid 1), \, \overline{Q}_{Y_0 \mid X}^\mathcal{T}(q \mid 1) \right].
\end{equation}
Moreover, the interior of the set in equation \eqref{eq:Y0quantileIdentifiedSet} equals the interior of the identified set. Finally, the proposition also holds if we replace $\mathcal{T}$ with $\mathcal{U}$ .
\end{proposition}

$\mathcal{T}$-independence of $Y_0$ from $X$ with $\mathcal{T} = [Q_{Y_0}(a),Q_{Y_0}(b)]$ is equivalent to the quantile independence assumptions $Q_{Y_0 \mid X}(\tau \mid x) = Q_{Y_0}(\tau)$ for all $\tau \in [a,b]$, by A\ref{assn:continuity}. The bounds \eqref{eq:Y0quantileIdentifiedSet} are also sharp for the function $Q_{Y_0 \mid X}(\cdot \mid 1)$ in a sense similar to that used in proposition \ref{prop:TauCDFbounds} in the appendix; we omit the formal statement for brevity. This functional sharpness delivers the following result.

\begin{corollary}\label{corr:ATT_Y0bounds}
Suppose the assumptions of proposition \ref{prop:quantileTrtBounds} hold. Let $\Exp( | Y_0 |) < \infty$. Then $\Exp(Y_0 \mid X=1)$ lies in the set
\[
	\left[ \underline{\Exp}^\mathcal{T}(Y_0 \mid X=1), \,\overline{\Exp}^\mathcal{T}(Y_0 \mid X=1) \right]
	\equiv
	\left[
		\int_0^1 \underline{Q}_{Y_0 \mid X}^\mathcal{T}(q \mid 1) \; dq, \,
		\int_0^1 \overline{Q}_{Y_0 \mid X}^\mathcal{T}(q \mid 1) \; dq
	\right].
\]
Moreover, the interior of this set equals the interior of the identified set for $\Exp(Y_0 \mid X=1)$. Finally, the corollary also holds if we replace $\mathcal{T}$ with $\mathcal{U}$.
\end{corollary}

By proposition \ref{prop:quantileTrtBounds} we have that $\mathcal{T}$-independence implies that $\text{QTT}(q)$ lies in the set
\[
	\left[ Q_{Y \mid X}(q \mid 1) - \overline{Q}_{Y_0 \mid X}^\mathcal{T}(q \mid 1), \;
		Q_{Y \mid X}(q \mid 1) - \underline{Q}_{Y_0 \mid X}^\mathcal{T}(q \mid 1) \right]
\]
and that the interior of this set equals the interior of the identified set for $\text{QTT}(q)$. Likewise for $\mathcal{U}$-independence. If $q \in \mathcal{T}$, then $\text{QTT}(q)$ is point identified under $\mathcal{T}$-independence; this follows immediately from our bound expressions above. This result---that a single quantile independence condition can be sufficient for point identifying a treatment effect---was shown by \cite{Chesher2003}. A similar result holds in the instrumental variables model of \cite{ChernozhukovHansen2005} and the LATE model of \cite{ImbensAngrist1994}. See the discussion around assumption 4 in section 1.4.3 of \cite{MellyWuthrich2017}.

By corollary \ref{corr:ATT_Y0bounds} we have that $\mathcal{T}$-independence implies that the ATT lies in the set
\[
	\left[ \Exp(Y \mid X=1) - \overline{\Exp}^\mathcal{T}(Y_0 \mid X=1), \;
		\Exp(Y \mid X=1) - \underline{\Exp}^\mathcal{T}(Y_0 \mid X=1) \right]
\]
and that the interior of this set equals the interior of the identified set for the ATT. Likewise for $\mathcal{U}$-independence. Furthermore, in appendix \ref{appendix:cdfbounds} we show that these ATT bounds have simple analytical expressions, obtained from integrating our closed form expressions for the bounds on $Q_{Y_0 \mid X}(q \mid 1)$.

\section{Empirical Illustration: The Effect of Child Soldiering on Wages}
\label{sec:empirical}

In this section we use our results to study the impact of relaxing the unconfoundedness assumption in an empirical study of the effects of child soldiering on wages. We do this using both the $\mathcal{T}$- and $\mathcal{U}$-independence relaxations of statistical independence. We find that the identified sets are significantly larger under $\mathcal{U}$-independence. This implies that the average value constraints imposed by $\mathcal{T}$-independence have substantial identifying power; recall that these constraints are the features of quantile independence that require the latent propensity score to be non-monotonic. In particular, the baseline empirical results are generally quite robust under the $\mathcal{T}$-independence relaxation of unconfoundedness, but not the under $\mathcal{U}$-independence relaxation. This difference highlights the importance of the choice of exogeneity assumptions in practice, and how researchers can use their beliefs about the form of latent selection to assist in this choice.

\subsection*{Background}

By collecting extensive survey data, \cite{BlattmanAnnan2010} study the impact of child abductions during a twenty year war in Uganda, where ``an unpopular rebel group has forcibly recruited tens of thousands of youth'' (page 882). Although they consider a variety of outcome variables, we focus on the impact of abduction on later life wages. 

The main identification problem is that selection into military service is typically non-random. They argue, however, that forced recruitment in Uganda led to conditional random assignment of military service. They condition on two variables: (1) Prewar household size, because larger households were less likely to be raided by small bands of rebels, and (2) Year of birth, because abduction levels varied over time, so that some youth ages were more likely to be abducted than others. Hence their identification strategy is based on unconfoundedness, conditioning on these two variables. Although their qualitative evidence supporting unconfoundedness is compelling, this assumption is still nonrefutable. We therefore use our results to assess the sensitivity of relaxing unconfoundedness on their empirical conclusions.

\subsection*{Sample Definition}

The data comes from phase 1 of SWAY, the Survey of War Affected Youth in northern Uganda (see \citealt{AnnanBlattmanHorton2006}). This phase has 1216 males born between 1975 and 1991. We look at the subsample of units who (1) have wage data available and (2) earned positive wages. This leaves us with 448 observations. Let $Y$ denote log wage. We define treatment $X$ to be an indicator that the person was \emph{not} abducted. We include the two covariates discussed above, age when surveyed and household size in 1996. We omit other covariates for simplicity.

Age has 17 support points, household size has 21 support points, and treatment has 2 support points. Hence there are 714 total conditioning variable cells, relative to our sample size of 448 observations. To ensure that our conditional quantile estimators are reasonably smooth in the quantile index, we collapse these conditioning variables into 8 cells. Specifically, we replace age with a binary indicator of whether one is above or below the median age. Likewise, we replace household size with a binary indicator of whether one lived in a household with above or below median household size. This gives 8 total conditioning variable cells, with approximately 55 observations each.

\subsection*{Baseline Analysis}

First we present estimates of conditional average treatment effects for the treated (CATT) and conditional quantile treatment effects for the treated (CQTTs), under the unconfoundedness assumption. For brevity we focus on the covariate cell $w$ = (age, household size) = (above median, above median). This group has the largest baseline effects of treatment, meaning that being abducted lowered their later life wages by the largest. Specifically, our estimate of the CATT for this group is 0.57. Our CQTT estimates are 0.67 for $\tau = 0.25$, 0.54 for $\tau = 0.5$, and 0.56 for $\tau = 0.75$. Note that our sample size is small, with 121 observations in this cell. We omit standard errors here because the purpose of this section is to illustrate the methods developed in our paper.

\subsection*{Sensitivity Analysis}

To check the robustness of these baseline point estimates to failure of unconfoundedness, we estimate identified sets for the CATT and CQTT using our results from section \ref{sec:treatmentEffects}. To highlight the importance of the choice of relaxation, we consider sets $\mathcal{T} = \mathcal{U}$. In this case, corollary \ref{prop:TauImpliesU} shows that $\mathcal{T}$-independence implies $\mathcal{U}$-independence. Hence identified sets using $\mathcal{T}$-independence must necessarily be weakly contained within identified sets using only $\mathcal{U}$-independence. We explore the magnitude of this difference in the data. Since $\mathcal{T}$-independence is simply $\mathcal{U}$-independence combined with some additional average value constraints, the difference between these identified sets tells us the identifying power of these additional average value constraints.

Specifically, we use the choice $\mathcal{T} = \mathcal{U} = [\delta,1-\delta]$ for $\delta \in [0,0.5]$. For $\delta = 0$, this choice corresponds to full conditional independence $Y_0 \indep X \mid W=w$ under both classes of assumptions. For $\delta = 0.5$, this choice corresponds to median independence for $\mathcal{T}$-independence, and no assumptions for $\mathcal{U}$-independence. Values of $\delta$ between 0 and $0.5$ yield conditional partial independence between $Y_0$ and $X$ for both classes of assumptions.

\begin{figure}[t]
\centering
\includegraphics[width=0.95\linewidth]{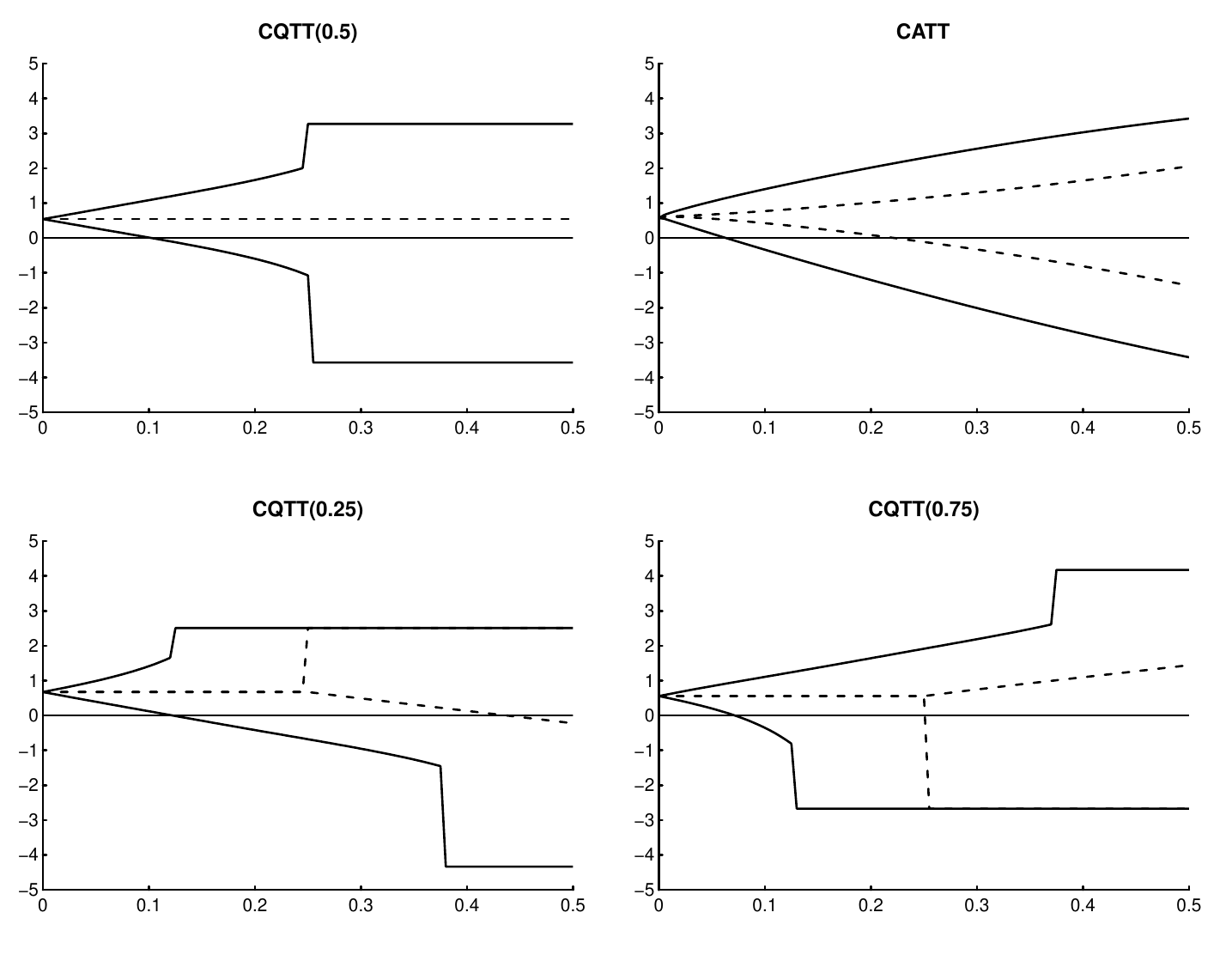}
\caption{Estimated identified sets for various parameters of interest, for $\mathcal{U} = \mathcal{T} = [\delta,1-\delta]$ with $\delta \in [0,0.5]$. Solid: $\mathcal{U}$-independence. Dashed: $\mathcal{T}$-independence. The horizontal axis shows values of $\delta$.}
\label{fig:empiricalPlots}
\end{figure}

Figure \ref{fig:empiricalPlots} shows estimated identified sets for both CATT and $\text{CQTT}(\tau)$ as $\delta$ varies from $0$ to $0.5$, and for $\tau \in \{ 0.25, 0.5, 0.75 \}$. These are sample analog estimates, where $\widehat{Q}_{Y \mid X,W}(\cdot \mid x,w)$ is estimated by inverting a kernel based estimate of $F_{Y \mid X,W}(\cdot \mid x,w)$. First consider the plot on the top left, which shows the estimated $\text{CQTT}(0.5)$ bounds. The dashed lines are the identified sets under $\mathcal{T}$-independence. Since median independence of $Y_0$ from $X$ conditional on $W=w$ is sufficient to point identify the conditional median $Q_{Y_0 \mid X,W}(0.5 \mid 1, w)$, median independence is also sufficient to point identify the CQTT at $0.5$. Hence the identified set is a singleton for all $\delta \in [0,0.5]$. This singleton equals 0.54, the baseline estimate. Next consider the solid lines. These are the estimated identified sets under $\mathcal{U}$-independence. When $\delta = 0.5$, $\mathcal{U}$-independence does not impose any constraints on the model, and hence we obtain the no assumption bounds, which are quite wide: $[-3.42, 3.42]$. If we decrease $\delta$ a small amount, thus making the $\mathcal{U}$-independence constraint nontrivial, the estimated identified set does not change. In fact, we can impose random assignment for about the middle 50\% of units (i.e., $\mathcal{U} = [0.25,0.75]$, or $\delta = 0.25$) and still we only obtain the no assumption bounds. Consequently, for intervals $\mathcal{T} \subseteq [0.25,0.75]$, the point identifying power of $\mathcal{T}$-independence is due solely to the constraint it imposes on the average value of the latent propensity score outside the interval $\mathcal{T}$, rather than the constraint that random assignment holds for units in the middle of the distribution of $Y_0$. 

Next define
\[
	\delta^\mathcal{U}_\text{bp}(\tau) = \sup \{ \delta \in [0,0.5] : \text{LB}^\mathcal{U}(\tau, \delta) \geq 0 \}
\]
where $\text{LB}^\mathcal{U}(\tau,\delta)$ is the lower bound of the identified set for $\text{CQTT}(\tau)$ under $\mathcal{U}$-independence with $\mathcal{U} = [\delta,1-\delta]$. Define $\delta^\mathcal{T}_\text{bp}(\tau)$ analogously. This value $\delta^\mathcal{U}_\text{bp}(\tau)$ is a breakdown point: It is the largest amount we can relax full independence while still being able to conclude that the treatment effect is nonnegative. For $\tau = 0.5$, the estimated breakdown point for $\text{CQTT}(0.5)$ is 0.103. Thus we can allow randomization to fail for about 20.6\% of units while still being able to conclude that $\text{CQTT}(0.5)$ is nonnegative. In contrast, as mentioned above, the breakdown point for $\mathcal{T}$-independence is always 0.5.

Next consider the lower two plots of figure \ref{fig:empiricalPlots}. These plots show estimated identified sets for $\text{CQTT}(0.25)$ on the left and $\text{CQTT}(0.75)$ on the right. There are two main differences between these plots and that of $\text{CQTT}(0.5)$: First, the $\mathcal{U}$-independence upper and lower bounds are not symmetric. Nonetheless, the qualitative robustness conclusions are similar. For example, $\widehat{\delta}^\mathcal{U}_\text{bp}(0.25)$ is 0.122 and $\widehat{\delta}^\mathcal{U}_\text{bp}(0.75)$ is 0.071. So conclusions about smaller quantiles are slightly more robust than conclusions about larger quantiles. Second, the $\mathcal{T}$-independence identified sets are no longer always singletons. In particular, we obtain non-singleton bounds when $\delta > 0.25$. However, conclusions under the $\mathcal{T}$-independence relaxation are substantially more robust than conclusions under the $\mathcal{U}$-independence relaxation. Specifically, $\widehat{\delta}^\mathcal{T}_\text{bp}(0.25)$ is 0.437. This is about 3.5 times as large as $\widehat{\delta}^\mathcal{U}_\text{bp}(0.25)$. Similarly, $\widehat{\delta}^\mathcal{T}_\text{bp}(0.75)$ is 0.25. This is also about 3.5 times as large as $\widehat{\delta}^\mathcal{U}_\text{bp}(0.75)$.

Finally consider the plot on the top right of figure \ref{fig:empiricalPlots}, which shows estimated identified sets for CATT. First consider the $\mathcal{T}$-independence relaxation, the dashed lines. The CATT is no longer point identified under median independence, or any set $\mathcal{T} \subsetneq (0,1)$ of quantile independence conditions; that is, the CATT is partially identified for all $\delta > 0$. Nonetheless, even median independence alone has substantial identifying power: For $\delta = 0.5$, the estimated identified set under median independence is $[-1.36, 2.06]$, whereas the no assumption bounds are $[-3.42, 3.42]$. Thus the width of the bounds has been cut in half. For $\delta > 0$, $\mathcal{U}$-independence has non-trivial identifying power, as shown in the solid lines. However, comparing the length of these bounds to the length to the $\mathcal{T}$-independence bounds, we see that imposing the average value constraint outside the interval $[\delta,1-\delta]$ again has substantial identifying power: the $\mathcal{T}$-independence bounds are anywhere from 50\% ($\delta = 0.5)$ to almost 100\% (arbitrarily small $\delta$) smaller than the $\mathcal{U}$-independence bounds. That is, the difference in lengths increases as we get closer to independence (as $\delta$ gets smaller). Thus conclusions about CATT are substantially more sensitive to small deviations from independence which do not impose the average value constraint outside the interval $[\delta,1-\delta]$, compared with small deviations which do impose that constraint. A second way to see this is to compare the breakdown points under the two relaxations. Define
\[
	\delta^\mathcal{U}_\text{bp} = \sup \{ \delta \in [0,0.5] : \text{LB}^\mathcal{U}(\delta) \geq 0 \}
\]
where $\text{LB}^\mathcal{U}(\delta)$ is the lower bound of the identified set for CATT under $\mathcal{U}$-independence with $\mathcal{U} = [\delta,1-\delta]$. Define $\delta^\mathcal{T}_\text{bp}$ analogously. As shown in the plots above, $\widehat{\delta}^\mathcal{U}_\text{bp} = 0.063$ while $\widehat{\delta}^\mathcal{T}_\text{bp} = 0.222$. Thus the breakdown point under $\mathcal{T}$-independence is again about 3.5 times as large as the breakdown point under $\mathcal{U}$-independence.

\subsection*{Empirical Conclusions}

In this section we used our identification results to study the robustness of conclusions about CATT and CQTTs to failures of unconfoundedness. Our baseline point estimates suggest that child abduction and forced military service has a negative effect on later life wages, for those children who were older when they were abducted and who came from larger households. This holds both on average (from the CATT) and across the distribution of treatment effects (as seen in the CQTTs). We then asked: How sensitive are these conclusions to failures of unconfoundedness? We saw that using the $\mathcal{T}$-independence relaxation, these conclusions are generally robust to large relaxations of unconfoundedness. However, using the $\mathcal{U}$-independence relaxation, these conclusions appear much more sensitive. As we earlier discussed, the difference arises from the additional average value constraints that $\mathcal{T}$-independence imposes. Those constraints are the features of quantile independence that require the latent propensity score to be non-monotonic. Thus it is critical to assess the plausibility of those additional constraints when deciding between these two forms of exogeneity assumptions to use for assessing sensitivity. 

In this empirical context, a monotonic latent propensity score arises when youths who have larger potential earnings when they're abducted (larger $Y_0$) are more likely to be abducted. If youths are targeted for abduction because of their innate or pre-existing skills, which would generally lead to large $Y_0$, then this would be a form of monotonic selection that would \emph{not} be allowed for by the $\mathcal{T}$-independence relaxation, but \emph{would} be allowed for by the $\mathcal{U}$-independence relaxation. So if we are concerned that unconfoundedness fails due to this kind of non-random selection, then $\mathcal{U}$-independence is a more appropriate choice for assessing sensitivity than $\mathcal{T}$-independence. Given this choice, the baseline results still hold under mild relaxations of unconfoundedness, since we saw that $\mathcal{U}$-independence breakdown points were generally around $\delta = 0.1$. But the baseline results no longer hold for larger relaxations; in this case, the data are inconclusive.

\section{The Treatment Selection Implications of a Roy Model} \label{sec:LatentSelectionModels}

As we emphasized, there is a direct mapping between exogeneity assumptions and the allowed forms of treatment selection. At one extreme, full independence assumes no selection at all of $X$ on $Y_x$, and therefore $p(y_x)$ is constant. On the other hand, weaker exogeneity assumptions allow for a class of deviations that one wishes to be robust against. Since this class is often not explicitly specified, we refer to such deviations as \emph{latent selection models}. Our main results in section \ref{sec:characterizationSection} characterize the set of latent selection models allowed by quantile and mean independence restrictions.

In this section, we consider a class of Roy Models and examine the relationship between their implied treatment selection functions and the exogeneity assumptions of section \ref{sec:characterizationSection}.  We discuss different assumptions on the economic primitives which lead these models to be either consistent or inconsistent with quantile or mean independence restrictions. We only consider single-agent models, but similar analyses can likely be done for multi-agent models.

Suppose we are again interested in identifying the average treatment effect for the treated parameter
\begin{align*}
	\text{ATT}
	&= \Exp(Y_1 - Y_0 \mid X=1) \\
	&= \Exp(Y \mid X=1) - \Exp(Y_0 \mid X=1).
\end{align*}
As in section \ref{sec:treatmentEffects}, its identification depends on our assumptions about the stochastic relationship between $X$ and $Y_0$. Suppose agents choose treatment to maximize their outcome:
\begin{equation}\label{eq:RoyModel}
	X = \indicator(Y_1 > Y_0).
\end{equation}
This is the classical Roy model (see \citealt{HeckmanVytlacil2007part1}). This assumption specifies how treatment $X$ relates to $Y_0$. Specifically, consider the latent propensity score
\begin{align*}
	p(y_0)	&\equiv \Prob(X=1 \mid Y_0 = y_0) \\
			&= \Prob(Y_1 > y_0 \mid Y_0 = y_0).
\end{align*}
The second line follows by our Roy model treatment choice assumption. Thus the shape of $p$ depends on the joint distribution of $(Y_1,Y_0)$. We classify these distributions into two possible cases, based on a concept called regression dependence (which is formally defined in definition \ref{def:regdependence} in appendix \ref{sec:ctsX}).
\begin{enumerate}
\item First suppose $(Y_1,Y_0)$ is such that $Y_1$ is regression dependent on $Y_0$. This implies that $p$ is monotonic. Corollary \ref{corr:monotonicPropensityScores} and proposition \ref{prop:mean-indep prop scores} therefore imply that no mean or quantile independence conditions of $Y_0$ on $X$ can hold unless $X \indep Y_0$. This occurs when $X$ is degenerate, as when treatment effects $Y_1-Y_0$ are constant, or more generally when $(Y_1 - Y_0) \indep Y_0$. In particular, any mean or quantile independence condition of $Y_0$ on $X$ rules out bivariate normally distributed $(Y_1,Y_0)$, again unless $X \indep Y_0$.

\item Next suppose $(Y_1,Y_0)$ is such that $Y_1$ is not regression dependent on $Y_0$. For example, let 
$Y_1 = Y_0 + \mu(Y_0) - \varepsilon$ where $\mu$ is a deterministic function and $\varepsilon \sim \normal(0,1)$, $\varepsilon \indep Y_0$. Then
\[
	p(y_0) = \Prob(X=1 \mid Y_0 = y_0) = \Phi[ \mu( y_0 ) ],
\]
where $\Phi$ is the standard normal cdf. If $\mu$ is non-monotonic then $p$ will also be non-monotonic. For this joint distribution of potential outcomes, the unit level treatment effects $Y_1-Y_0$ conditional on the baseline outcome $Y_0=y_0$ are distributed $\normal( \mu(y_0), 1)$. Hence non-monotonicity of $\mu$ implies that the mean of this distribution of treatment effects is not monotonic. For instance, suppose the outcome is earnings and treatment is completing college. Let
\begin{align*}
	\mu(y_0) &> 0 \qquad\text{if $y_0 \in (\alpha,\beta)$} \\
	\mu(y_0) &\leq 0 \qquad\text{if $y_0 \in (-\infty,\alpha] \cup [\beta,\infty)$}
\end{align*}
for $-\infty < \alpha < \beta < \infty$. Then people with sufficiently small or sufficiently large earnings when they do not complete college do not benefit from completing college, on average. People with moderate earnings when they do not complete college, on the other hand, do typically benefit from completing college. This kind of joint distribution of potential outcomes combined with the Roy model assumption \eqref{eq:RoyModel} on treatment selection produces non-monotonic latent propensity scores.

\medskip

We just gave one example joint distribution of $(Y_1,Y_0)$ where regression dependence fails. More generally, theorem 5.2.10 on page 196 of \cite{Nelsen2006} characterizes the set of copulas for which $Y_1$ is regression dependent on $Y_0$, when both are continuously distributed. This result therefore also tells us the set of copulas where $Y_1$ is \emph{not} regression dependent on $Y_0$. Among these copulas, $\mathcal{T}$-independence (or, analogously, mean independence) of $Y_0$ from $X$ will specify a further subset of allowed dependence structures. The precise set is given by all copulas which lead to latent propensity scores that satisfy the average value constraint. 
\end{enumerate}
Whether one of these cases is plausible depends on the specific application at hand. For example, Heckman, Smith, and Clements' \citeyearpar{HeckmanSmithClements1997} study the Job Training Partnership Act (JTPA). They find that ``plausible impact distributions require high measures of positive dependence [of $Y_1$ on $Y_0$]'' (page 506). This suggests that case 1 is more relevant for their data, and hence it may be unlikely that any quantile or mean independence holds in their setting.

\section{Conclusion}

In this paper we gave several results to help researchers assess the plausibility of quantile and mean independence assumptions on structural unobservables like potential outcomes. Keep in mind, however, that when doing identification analysis it is not necessary to choose a single exogeneity assumption. For example, researchers may want to consider a variety of exogeneity assumptions in this step, as part of a sensitivity analysis. We illustrated this in sections \ref{sec:treatmentEffects} and \ref{sec:empirical}. The choice of which exogeneity assumptions to consider is still determined by considering the kinds of treatment selection we want to allow for, as discussed in section \ref{sec:whyRelaxIndependence}. Conversely, there may be situations where researchers do not find it plausible to impose \emph{any} kind of exogeneity assumption. In this case we often can still learn something about the parameters of interest, as in the classical no assumption bounds of \cite{Manski1990}. In this paper we focused on the case where the researcher does want to impose some kind of exogeneity assumption, however. In this case, we hope that our results can help researchers better select the most appropriate exogeneity assumptions for their settings.

\singlespacing
\bibliographystyle{econometrica}
\bibliography{QI_paper}

\normalsize
\allowdisplaybreaks

\appendix

\section{Structural and Reduced Form Unobservables}\label{sec:twoKindsOfUnobs}

The approach we recommend in section \ref{sec:whyRelaxIndependence} begins by distinguishing between two kinds of unobservables: (1) Structural unobservables and (2) Reduced form unobservables. In this appendix we discuss a specific example to clarify the distinction between these two variables. We also use this example to discuss the difference between exogeneity assumptions involving reduced form unobservables and those involving structural unobservables.

\subsubsection*{Example: The Binary Response Random Coefficients Model}

Let $X$ be a scalar observed random variable. Let $Y^*(x)$ denote a latent potential outcome for $x \in \R$. Let $Y^* = Y^*(X)$ denote realized latent potential outcomes. Let
\[
	Y(x) = \indicator[ Y^*(x) \geq 0 ]
\]
denote the usual potential outcomes. Let $Y = Y(X)$ denote our observed outcome. Suppose latent potential outcomes satisfy the linear random coefficient model
\begin{equation}\label{eq:deepStructuralLinearRC}
	Y^*(x) = A + B x,
\end{equation}
where $A$ and $B$ are structural unobserved random variables. Suppose we impose the following constraint on the tails of $A$, $B$, and $X$.

\begin{misspecAssump}\label{assump:finiteMeansRC}
$\Exp(A)$, $\Exp(B)$, and $\Exp(X)$ exist and are finite.
\end{misspecAssump}

In this model, $(A,B)$ are structural unobservables. An exogeneity assumption about the relationship between $(A,B)$ and $X$ is therefore a statement about treatment selection: How does assigned treatment depend on these structural unobservables? This is the kind of exogeneity assumption we focus on in this paper. In this appendix however, we will discuss a second kind of exogeneity condition, which constrains the relationship between treatment and a reduced form unobservable. This is a \emph{derived} exogeneity condition: It is not directly imposed but rather is a consequence of a choice of (1) an exogeneity assumption involving the structural unobservables and (2) the functional form of the reduced form unobservables.

To illustrate this kind of derived exogeneity condition, we'll assume treatment is randomly assigned.

\begin{misspecAssump}\label{assump:indepRandomCoeffs}
$(A,B) \indep X$.
\end{misspecAssump}

Next write the equation for realized latent potential outcomes as
\begin{align}\label{eq:quasiLinearRC}
	Y^*(X)
		&= \Exp(A) + \Exp(B) X + \big( [A - \Exp(A)] + [B - \Exp(B)] X \big) \notag \\
		&\equiv \Exp(A) + \Exp(B) X + V. 
\end{align}
$V$ is a reduced form unobservable. It is a function of the structural unobservables $(A,B)$ as well as the realized treatment $X$. Consequently, it is not invariant to changes in the distribution of $X$. The coefficients $\Exp(A)$ and $\Exp(B)$ in equation \eqref{eq:quasiLinearRC} do have structural interpretations, however. %

By definition, $V$ depends on $X$, and so typically $V$ is not independent of $X$. Nonetheless, we can derive some restrictions on the distribution of $V \mid X$. Specifically, suppose we also make the following assumption.

\begin{misspecAssump}\label{assump:symmetricPotentialOutcomes}
$Y^*(x)$ is symmetrically distributed about $\Exp[Y^*(x)]$ for all $x \in \supp(X)$.
\end{misspecAssump}

Under this additional assumption, \citet[page 220]{Manski1975} showed the following result; also see \citet[pages 247--249]{Manski1977} and \citet[pages 1007--1008]{Fox2007}. 

\begin{proposition}\label{prop:randomCoefsGiveMedianIndep}
Suppose B\ref{assump:finiteMeansRC}, B\ref{assump:indepRandomCoeffs}, and B\ref{assump:symmetricPotentialOutcomes} hold. Then $\Prob(V \leq 0 \mid X=x) = \Prob(V \leq 0) = 0.5$ for all $x \in \supp(X)$. That is, $V$ is median independent of $X$.
\end{proposition}

Thus we have derived a median independence restriction on the reduced form unobservable $V$ as a consequence of (1) the definition of $V$ and (2) the exogeneity assumption about the relationship between $X$ and the structural unobservables.

\subsubsection*{Discussion}

In this example there are two kinds of unobservables: The structural unobservables $(A,B)$ and the reduced form unobservable $V$. We made an exogeneity assumption about the relationship between $(A,B)$ and $X$ based on our beliefs about treatment assignment. We then \emph{derived} an exogeneity condition on the relationship between the reduced form unobservable $V$ and $X$. More generally, for researchers interested in choosing an exogeneity assumption that relates treatment to reduced form unobservables, our recommendation is that this assumption be derived from a more primitive exogeneity assumption about the structural unobservables, as in the example. The analysis in sections \ref{sec:whyRelaxIndependence}--\ref{sec:LatentSelectionModels} of our paper can then be used to assess the plausibility of these more primitive exogeneity assumptions.

Researchers may sometimes prefer to make an exogeneity assumption on the reduced form unobservable directly, without directly deriving it from a more primitive model like we did above. There are a few concerns with this approach, however:
\begin{enumerate}
\item As \cite{Angrist2001} argues, we often care about parameters like average structural functions (ASFs) and average treatment effects (ATEs). In nonseparable models like the example above, however, the ASF depends on the distribution of the structural unobservables. In particular, for that example,
\begin{align*}
	\text{ASF}(x)
		&= \Exp[Y(x)] \\
		&= \Prob_{A,B}(A + B x \geq 0) \\
		&\neq \Prob_V( \Exp(A) + \Exp(B) x + V \geq 0).
\end{align*}
The true ASF does \emph{not} generally equal the parameter that you would compute if you worked with equation \eqref{eq:quasiLinearRC}, but incorrectly treated $V$ as a structural unobservable in an ASF calculation. Thus, in nonseparable models, it is generally not possible to avoid working with structural unobservables if we are interested in ASFs and ATEs.

For example, among many other derivations, \cite{Torgovitsky2015} computes identified sets for ASFs and ATEs in a binary response model with constant coefficients and median independence. For these identified sets to have the correct interpretation, the unobservable $V$ in his model must be structural, rather than a reduced form. Consequently, the fact that we are only imposing median independence of these unobservables from treatment implies that we are concerned with a particular kind of non-random treatment assignment. Our results in section \ref{sec:characterizationSection} characterize the kind of non-random treatment assignment consistent with median independence assumptions.

\item If we only work with the reduced form unobservables, then we might be ignoring a lot of useful information. Consider the example again: Relative to the assumption that $V$ is median independent of $X$, the stronger assumptions B\ref{assump:indepRandomCoeffs} and B\ref{assump:symmetricPotentialOutcomes} have potentially different implications for falsification, identification, rates of convergence, and efficiency. For example, in the model we know that the ASF is point identified:
\[
	\Exp[Y(x)] = \Prob(Y=1 \mid X=x).
\]
But if we only impose median independence of $V$ from $X$ then the ASF is generally only partially identified.
\end{enumerate}
For these reasons, we recommend working directly with the structural unobservables. This does not require that researchers make strong assumptions like statistical independence on these structural unobservables, however. Instead, they can use the methods discussed in this paper to think about the form of exogeneity you want to impose on the relationship between the structural unobservables and the treatment variables.

\section{Characterizing Quantile Independence with Non-Binary Treatments}\label{sec:multivalTreat}

In this section we generalize our characterization results to allow treatment $X$ to be non-binary. For any logically possible value of treatment $x$, let $Y(x)$ denote the potential outcomes. Although our results extend to general potential outcomes, they are arguably more natural for the simpler nonseparable model
\[
	Y(x) = m(x,U)
\]
where $U$ is continuously distributed scalar heterogeneity and $m$ is an unknown, potentially nonseparable function that is strictly increasing in $U$. This model allows for the impact of changes in treatment values to arbitrarily depend on $U$. It was studied in \cite{Matzkin2003}, \cite{Imbens2007}, and \cite{Torgovitsky2015a}, among others. Note that it includes the classical linear model for potential outcomes as a special case.

\subsection{Extension to Continuous Treatments}\label{sec:ctsX}

First we extend our main characterization result (theorem \ref{thm:AvgValueCharacterization}) to continuous treatments $X$. In this setting, it is useful to work with this equivalent representation of $\tau$-cdf independence condition \eqref{cdfIndependence}:
\begin{align*}
	F_{U \mid X}(\tau \mid x) &= F_{U}(\tau)
\end{align*}
for all $x \in \supp(X)$. This equivalence exists since there is an invertible mapping between $Y_x$ and $U$, and by quantile invariance. Since the structural function $m$ and the distribution of $U$ are not separately identified, it is common to normalize $U \sim \text{Unif}[0,1]$. We also use this normalization to simplify the statements and proofs of the results, but all of our results can be extended to the case where $U$ is not normalized.

As in the binary $X$ case, our results show that the deviations from independence allowed by quantile independence require a kind of non-monotonic selection on unobservables. We start by giving the analog to theorem \ref{thm:AvgValueCharacterization} in the continuous $X$ case.

\begin{theorem}[Average value characterization]
\label{thm:AvgValueCharacterization_ctsX}
Suppose $X$ and $U$ are continuously distributed; normalize $U \sim \text{Unif}[0,1]$. Then $U$ is $\mathcal{T}$-independent of $X$ if and only if
\begin{equation}\label{eq:averageValueCondition_main_ctsX}
	\Exp \big( \Prob(X > x \mid U ) \mid  U \in (t_1,t_2) \big) = \Prob(X > x)
	\quad \text{for all $x \in \supp(X)$}
\end{equation}
for all $t_1, t_2 \in\mathcal{T} \cup \{ 0,1 \}$ with $t_1 < t_2$.
\end{theorem}

The interpretation of equation \eqref{eq:averageValueCondition_main_ctsX} is similar to the binary $X$ case: $\mathcal{T}$-independence holds if and only if, for each possible level of treatment $x$, and for each interval with endpoints in $\mathcal{T} \cup \{ 0,1 \}$ the average value of the conditional probability of receiving treatment larger than $x$ given the unobservable equals the overall unconditional probability of receiving treatment larger than $x$. Notice that, by adding $-1$ to each side of equation \eqref{eq:averageValueCondition_main_ctsX}, this constraint can equivalently be seen as a constraint on the conditional cdf $F_{X \mid U}$.

As in the binary $X$ case, the constraint \eqref{eq:averageValueCondition_main_ctsX} imposes a non-monotonicity condition.

\begin{corollary}\label{corr:noRegDependence}
Suppose $X$ and $U$ are continuously distributed. Suppose there is some $x \in \supp(X)$ such that $\Prob(X > x \mid U=u)$ is weakly monotonic and not constant over $u \in \text{int}[\supp(U)]$. Then, for all $\tau \in \supp(U)$, $U$ is not $\tau$-cdf independent of $X$.
\end{corollary}

For example, suppose $X$ is level of education, $x$ is completing college, and $U$ is ability. Then any nontrivial $\mathcal{T}$-independence condition implies that at some point increasing ability lowers the probability of getting more than a college education (or that it is constant in ability). 

The monotonicity condition of corollary \ref{corr:noRegDependence} dates back to \cite{Tukey1958} and \cite{Lehmann1966}, who give the following definition.

\begin{definition}\label{def:regdependence}
Say $X$ is \emph{positively [negatively] regression dependent} on $U$ if $\Prob(X > x \mid U=u)$ is weakly increasing [decreasing] in $u$, for all $x \in \R$. Say $X$ is \emph{regression dependent} on $U$ if it is either positively or negatively regression dependent on $U$.
\end{definition}
Regression dependence is also known as \emph{stochastic monotonicity}, since it is equivalent to the set of cdfs $\{ F_{X \mid U}(\cdot \mid u) : u \in \supp(U) \}$ being either increasing or decreasing in the first order stochastic dominance ordering. Thus corollary \ref{corr:noRegDependence} states that we cannot simultaneously have quantile independence of $U$ on $X$ and regression dependence of $X$ on $U$, except when $X \indep U$.

\cite{LehmannRomano2005} call positive regression dependence an ``intuitive meaning of positive dependence''. To support this claim, \cite{Lehmann1966} gave the following simple sufficient conditions for regression dependence: If one can write $X = \pi_0 + \pi_1 U + V$ where $\pi_0$ and $\pi_1$ are constants and $V$ is a random variable independent of $U$, then $X$ is regression dependent on $U$ if $\pi_1 \neq 0$. In particular, if $X$ and $U$ are jointly normally distributed then they are regression dependent so long as they have nonzero correlation. While these are special cases, theorem 5.2.10 on page 196 of \cite{Nelsen2006} provides a general characterization of regression dependence in terms of the copula between $X$ and $U$, when both variables are continuous. In particular, if $C_{X,U}(x,u)$ is the copula for $(X,U)$, $X$ is regression dependent on $U$ if and only if $C_{X,U}(x,\cdot)$ is concave for any $x \in [0,1]$.

Several papers in econometrics have previously used stochastic monotonicity assumptions for identification. \cite{BlundellEtAl2007} study the classic problem of identifying the distribution of potential wages, given that wages are only observed for workers. Following \cite{ManskiPepper2000}, they argue that stochastic monotonicity assumptions are often plausible. They specifically consider stochastic monotonicity of wages on labor force participation status, as well as stochastic monotonicity of wages on an instrument. They furthermore provide a detailed analysis of when stochastic monotonicity assumptions may not be plausible. 

Corollary \ref{corr:noRegDependence} shows that any assumption of $\mathcal{T}$-independence of $U$ on $X$ rules out stochastic monotonicity of $X$ on $U$. Thus, if one wants to allow for a class of deviations from independence which includes stochastically monotonic selection, assumptions of quantile independence of $U$ on $X$ should not be used. Conversely, if one makes a quantile independence assumption of $U$ on $X$, one should argue why stochastically non-monotonic selection models are the deviations of interest. We discuss these issues further in section \ref{sec:LatentSelectionModels}.

\subsection{Extension to Discrete Treatments}\label{sec:discreteX}

Next we consider discrete multi-valued $X$.

\begin{theorem}[Average value characterization]\label{thm:AvgValueCharacterization_discreteX}
Suppose $X$ is discrete. Suppose $U$ is continuously distributed; normalize $U \sim \text{Unif}[0,1]$. Then $U$ is $\mathcal{T}$-independent of $X$ if and only if
\begin{equation}\label{eq:averageValueCondition_main_discreteX}
	\Exp \big( \Prob(X=x \mid U)  \mid  U \in (t_1,t_2) \big) = \Prob(X=x)
	\qquad \text{for all $x \in \supp(X)$}
\end{equation}
for all $t_1, t_2 \in\mathcal{T} \cup \{ 0,1 \} $ with $t_1 < t_2$.
\end{theorem}

This result has a similar interpretation as our previous results for binary $X$ and continuous $X$. First, we have the following corollary.

\begin{corollary}\label{corr:discreteXnotMonotone}
Suppose $X$ is discrete. Suppose $U$ is continuously distributed. Suppose there is some $x \in \supp(X)$ such that $\Prob(X \geq x \mid U=u)$ is weakly monotonic and not constant over $u \in \text{int}[\supp(U)]$. Then, for all $\tau \in \supp(U)$, $U$ is not $\tau$-cdf independent of $X$.
\end{corollary}

The interpretation is analogous to corollary \ref{corr:noRegDependence}. Second, all of the interpretations given in section \ref{sec:characterizationSection} apply to the probabilities $\Prob(X = x \mid U=u)$ for $x \in \supp(X)$. In particular, these conditional probabilities must be non-monotonic. For example, suppose $\supp(X) = \{ x_1,\ldots,x_K \}$ is finite. This non-monotonicity result is primarily relevant for the lowest treatment level ($x=x_1$) and the highest treatment level ($x=x_K$), since non-monotonicity of the middle probabilities would be implied, for example, by a simple ordered threshold crossing model, like $X = x_k$ if $\alpha_k \leq U \leq \alpha_{k+1}$ for constants $\alpha_k \leq \alpha_{k+1}$, $k \in \{ 1,\ldots,K \}$.

\subsection{Generalizing $\mathcal{U}$-independence}

Finally, the following result extends corollary \ref{prop:TauImpliesU} to allow $X$ to be discrete multi-valued or continuous.

\begin{corollary}\label{prop:TauImpliesU_ctsX}
Suppose $U$ is continuously distributed; normalize $U \sim \text{Unif}[0,1]$. Let $\mathcal{T} = [a,b]\subseteq [0,1]$. Then $\mathcal{T}$-independence of $U$ from $X$ implies $F_{X \mid U}(x \mid u) = F_X(x)$ for all $x \in \R$ and almost all $u \in \mathcal{T}$.
\end{corollary}

Although we focused on binary $X$ in section \ref{sec:treatmentEffectBounds}, this corollary suggests that we can generalize our definition of $\mathcal{U}$-independence to allow multi-valued or continuous treatments by specifying $F_{X \mid U}(x \mid u) = F_X(x)$ for all $x \in \R$ and almost all $u \in \mathcal{U}$.

\section{Definitions of the Bound Functions}\label{appendix:cdfbounds}

In this appendix we provide the precise functional forms for the cdf bounds of proposition \ref{prop:TauCDFbounds} and the conditional mean bounds of corollary \ref{corr:ATT_Y0bounds}.

\subsection*{The cdf bounds}

The $\mathcal{T}$-independence bounds are as follows:
\[
	\overline{F}_{U \mid X}^\mathcal{T}(u \mid x)
		=
		\begin{cases}
			\dfrac{F_U(u)}{p_x}
				&\text{if $u \leq  Q_U(p_x F_U(a))$} \\
			F_U(a)
				&\text{if $ Q_U(p_x F_U(a)) \leq u \leq a$} \\
			F_U(u)
				& \text{if $a \leq u \leq b$} \\
			\dfrac{F_U(u) - F_U(b)}{p_x} + F_U(b)
				&\text{if $b \leq u \leq Q_{U}(p_x + F_U(b)(1-p_x))$} \\
			1
				&\text{if $Q_{U}(p_x + F_U(b)(1-p_x))  \leq u$}
		\end{cases}
\]
and
\[
	\underline{F}_{U \mid X}^\mathcal{T}(u \mid x)
		=
		\begin{cases}
			0
				&\text{if $ u \leq  Q_U((1-p_x) F_U(a))$} \\
			\dfrac{F_U(u)-F_U(a)}{p_x} + F_U(a)
				&\text{if $ Q_U((1-p_x) F_U(a)) \leq u \leq a$}\\
			F_U(u)
				& \text{if $ a \leq u \leq b$}\\
			F_U(b)
				&\text{if $b \leq u \leq Q_U(p_x F_U(b) + (1-p_x)) $} \\
			\dfrac{F_U(u)-1}{p_x} + 1
				&\text{if $Q_U(p_x F_U(b) + (1-p_x))  \leq u$}.
		\end{cases}
\]
For $\mathcal{U}$-independence, first consider the lower bound. There are two separate cases. First, if $(1-(F_U(b)-F_U(a)))(1-p_x) \leq F_U(a)$,
\begin{multline*}
	\underline{F}_{U \mid X}^\mathcal{U}(u \mid x)
	= \\
	\begin{cases}
		0 
			&\text{ for $u \leq Q_U((1-(F_U(b)-F_U(a)))(1-p_x))$} \\
		\dfrac{F_U(u) - (1-(F_U(b)-F_U(a)))(1-p_x)}{p_x}
			&\text{ for $u\in[Q_U((1-(F_U(b)-F_U(a)))(1-p_x)),a]$} \\
		\dfrac{(F_U(b)-1)(1-p_x) }{p_x} + F_U(u)
			&\text{ for $u \in [a,b]$} \\
		\dfrac{F_U(u)-1}{p_x} +1
			&\text{ for $u \geq b$}.
	\end{cases}
\end{multline*}
Second, if $(1-(F_U(b)-F_U(a)))(1-p_x) \geq F_U(a)$,
\[
	\underline{F}_{U \mid X}^\mathcal{U}(u \mid x) =
	\begin{cases}
		0
			&\text{ for $u\leq a$} \\
		F_U(u)-F_U(a)
			&\text{ for $u\in[a,b]$} \\
		F_U(b)-F_U(a)
			&\text{ for $u\in[b,Q_U(p_x(F_U(b)-F_U(a)) + 1-p_x)]$} \\
		\dfrac{F_U(u)-1}{p_x} +1
			&\text{ for $u \geq Q_U(p_x(F_U(b)-F_U(a)) + 1-p_x)$.}
	\end{cases}
\]
Next consider the upper bound. Again, there are two separate cases. First, if $(1-(F_U(b)-F_U(a)))p_x \leq F_U(a)$,
\[
	\overline{F}_{U \mid X}^\mathcal{U}(u \mid x) =
	\begin{cases}
		\dfrac{F_U(u)}{p_x}
			&\text{ for $u \leq Q_U((1-(F_U(b) - F_U(a)))p_x)$} \\
		1-(F_U(b)-F_U(a))
			&\text{ for $u \in [Q_U((1-(F_U(b) - F_U(a)))p_x),a]$} \\
		1-(F_U(b)-F_U(u))
			&\text{ for $u \in [a,b]$} \\
		1
			&\text{ for $u \geq b$.}
	\end{cases}
\]
Second, if $(1-(F_U(b)-F_U(a)))p_x \geq F_U(a)$,
\begin{multline*}
	\overline{F}_{U \mid X}^\mathcal{U}(u \mid x) = \\
	\begin{cases}
		\dfrac{F_U(u)}{p_x}
			&\text{ for $u \leq a$} \\
		\dfrac{F_U(a)}{p_x} + F_U(u)-F_U(a)
			&\text{ for $u \in [a,b]$} \\
		\dfrac{(F_U(b)-F_U(a))(p_x-1) +F_U(u)}{p_x}
			&\text{ for $u \in [b,Q_U((F_U(b)-F_U(a))(1-p_x)+ p_x)]$} \\
		1
			&\text{ for $u \geq Q_U((F_U(b)-F_U(a))(1-p_x)+ p_x)$.}
	\end{cases}
\end{multline*}

\subsection*{The conditional mean bounds}

By integrating the quantile bounds as in the statement of corollary \ref{corr:ATT_Y0bounds}, we obtain the bounds on $\Exp(Y_0 \mid X=1)$. We provide the explicit form of these bounds but omit the derivations for brevity. For $\mathcal{T}$-independence,
\[
	\overline{\Exp}^\mathcal{T}(Y_0 \mid X=1) =
	a Q_{Y \mid X}(a \mid 0) + \int_a^b Q_{Y \mid X}(\tau \mid 0) \; d\tau + (1-b)Q_{Y \mid X}(1 \mid 0)
\]
and
\[
	\underline{\Exp}^\mathcal{T}(Y_0 \mid X=1) =
	a Q_{Y \mid X}(0 \mid 0) + \int_a^b Q_{Y \mid X}(\tau \mid 0) \; d\tau + (1-b)Q_{Y \mid X}(b \mid 0).
\]
For $\mathcal{U}$-independence, first consider the lower bound. There are two cases. If $(1-(b-a))p_1 \leq a$,
\[
	\underline{\Exp}^\mathcal{U}(Y_0 \mid X=1)
	= (1-(b-a))Q_{Y \mid X}(0 \mid 0)
	+ \int_{1-(b-a)	+ \frac{b-1}{p_0}}^{1 + \frac{b-1}{p_0}} Q_{Y \mid X}(\tau \mid 0) \; d\tau.
\]
If $(1-(b-a))p_1 \geq a$,
\[
	\underline{\Exp}^\mathcal{U}(Y_0 \mid X=1) = \frac{a}{p_1} Q_{Y \mid X}(0 \mid 0)
	+ \int_0^{b-a}Q_{Y \mid X}(\tau \mid 0) \; d\tau
	+ \left(1-(b-a) - \frac{a}{p_1} \right) Q_{Y \mid X}(b-a \mid 0).
\]
Next consider the upper bound. If $(1-(b-a))p_0 \leq a$,
\begin{multline*}
	\overline{\Exp}^\mathcal{U}(Y_0 \mid X=1)
	= \\
	Q_{Y \mid X}(1-(b-a) \mid 0) \left( 1-(b-a) - \frac{1-b}{p_1} \right) + \int_{1-(b-a)}^1Q_{Y \mid X}(\tau\mid 0) \; d\tau + Q_{Y \mid X}(1 \mid 0)\frac{1-b}{p_1}.
\end{multline*}
If $(1-(b-a))p_0 \geq a$,
\[
	\overline{\Exp}^\mathcal{U}(Y_0 \mid X=1)
	= \int_{\frac{a}{p_0}}^{b-a + \frac{a}{p_0}}Q_{Y \mid X}(\tau\mid 0) \; d\tau + (1-(b-a))Q_{Y \mid X}(1 \mid 0).
\]

\section{Proofs}\label{sec:proofs}

We use this simple lemma in some of our proofs.

\begin{lemma}\label{lemma:continuity}
Suppose $U$ is continuously distributed. Suppose $X$ is discrete. Then $F_{U \mid X}(\cdot \mid x)$ is a continuous function for all $x \in \supp(X)$.
\end{lemma}

\begin{proof}[Proof of lemma \ref{lemma:continuity}]
Suppose by way of contradiction that there is some $x^* \in \supp(X)$ such that $F_{U \mid X}(\cdot \mid x^*)$ is not continuous at some point $u^*$. Since cdfs are right-continuous, we must have $\lim_{ u \nearrow u^*} F_{U \mid X}(u \mid x^*) < F_{U \mid X}(u^* \mid x^*)$. This implies $\Prob(U = u^* \mid X = x^*) > 0$. Therefore
\begin{align*}
	0	&= \Prob(U = u^*)
			&\text{by $U$ continuously distributed} \\
		&= \sum_{x \in \supp(X)} \Prob( U = u^* \mid X = x) \Prob(X=x)
			&\text{by the law of total probability} \\
		&\geq \Prob( U = u^* \mid X = x^*) \Prob(X=x^*) \\
		&> 0.
\end{align*}
This is a contradiction.
\end{proof}

\subsection{Proofs for section \ref{sec:characterizationSection}}

\begin{proof}[Proof of theorem \ref{thm:AvgValueCharacterization}]
By the law of iterated expectations, this is equivalent to showing that $\Prob(X=1 \mid Y_x \in (t_1,t_2)) = \Prob(X=1)$ for all  $t_1, t_2 \in\mathcal{T} \cup \{ \underline{y}_x, \overline{y}_x \}$ with $t_1 < t_2$.

\begin{itemize}
\item[($\Rightarrow$)] Suppose $Y_x$ is $\mathcal{T}$-independent of $X$. Let $t_1, t_2 \in\mathcal{T} \cup \{ \underline{y}_x, \overline{y}_x \}$ with $t_1 < t_2$. Then, 
\begin{align*}
	\Prob(X=1 \mid Y_x \in (t_1,t_2))
		&= \frac{\Prob(Y_x \in (t_1,t_2) \mid X=1) \Prob(X=1)}{\Prob(Y_x \in (t_1,t_2))} \\
		&= \frac{(\Prob(Y_x < t_2 \mid X=1) - \Prob(Y_x \leq t_1 \mid X=1)) \Prob(X=x)}{\Prob(Y_x < t_2) - \Prob(Y_x \leq t_1)} \\
    		&= \frac{(\Prob(Y_x \leq t_2 \mid X=1) - \Prob(Y_x \leq t_1 \mid X=1)) \Prob(X=x)}{\Prob(Y_x\leq t_2) - \Prob(Y_x \leq t_1)} \\
    		&= \Prob(X=1).
\end{align*}
The third equality follows since $Y_x \mid X$ is continuously distributed, which itself follows by $X$ being discretely distributed and lemma \ref{lemma:continuity}. The fifth line follows from $\mathcal{T}$-independence, which is equivalent to $\Prob(Y_x \leq t \mid X=1) = \Prob(Y_x \leq t)$ for $t \in \mathcal{T}$.

\item[($\Leftarrow$)] Suppose that,
\[
	\Prob(X=1 \mid Y_x \in (t_1,t_2)) = \Prob(X=1)
\]
for all $t_1, t_2 \in\mathcal{T} \cup \{ \underline{y}_x, \overline{y}_x \}$ with $t_1 < t_2$. Then,
\begin{align*}
	\Prob(Y_x \in (t_1,t_2) \mid X=1)
		&= \frac{\Prob(X=1 \mid Y_x \in (t_1,t_2)) \Prob(Y_x \in (t_1,t_2))}{\Prob(X=1)} \\
		&= \frac{\Prob(X=1) \Prob(Y_x \in (t_1,t_2))}{\Prob(X=1)} \\
    		&= \Prob(Y_x \in (t_1,t_2)).
\end{align*}
The second line follows by assumption. Let $t_1 = \underline{y}_x$. Then, using lemma \ref{lemma:continuity}, $F_{Y_x \mid X}(t_2 \mid 1) = \Prob(Y_x \in (\underline{y}_x,t_2) \mid X=1) = \Prob(Y_x \in (\underline{y}_x,t_2)) = F_{Y_x}(t_2)$. Thus $\mathcal{T}$-cdf independence holds.
\end{itemize}
\end{proof}

\begin{proof}[Proof of corollary \ref{corr:monotonicPropensityScores}]
Without loss of generality, suppose $p$ is weakly increasing. Fix $y \in (\underline{y}_x,\overline{y}_x)$. Then
\begin{align*}
	\Exp(p(Y_x) \mid Y_x \in (\underline{y}_x, y)) &\leq p(y)
\end{align*}
by monotonicity of $p$. Similarly, 
\begin{align*}
	\Exp(p(Y_x) \mid Y_x \in (y,\overline{y}_x)) &\geq p(y).
\end{align*}
Therefore 
\[
	\Exp(p(Y_x) \mid Y_x \in (\underline{y}_x, y)) \leq p(y) \leq \Exp(p(Y_x) \mid Y_x \in (y,\overline{y}_x)).
\]
Suppose these hold with equality. Then
\[
	\Exp(p(Y_x) - p(y) \mid Y_x \in (\underline{y}_x, y)) = 0,
\]
which implies that $p(y_x) = p(y)$ for all $y_x < y$ by $p(y_x) - p(y) \leq 0$ for $y_x < y$. Likewise, we would also have that
\[
	\Exp(p(Y_x) - p(y) \mid Y_x \in (y,\overline{y}_x)) = 0,
\]
which implies that $p(y_x) = p(y)$ for all $y_x > y$ by $p(y_x) - p(y) \geq 0$ for $y_x > y$. 

Therefore $p(y_x) = p(y)$ for all $y_x$, which contradicts our assumption that $p(y_x)$ is not constant. Hence
\[
	\Exp(p(Y_x) \mid Y_x \in (\underline{y}_x, y)) < \Exp(p(Y_x) \mid Y_x \in (y,\overline{y}_x)).
\]
Because $y$ was arbitrary, we see that $y$-cdf independence cannot hold for any $y \in (\underline{y}_x,\overline{y}_x)$, by theorem \ref{thm:AvgValueCharacterization}.
\end{proof}

\begin{proof}[Proof of corollary \ref{corr:SignChanges}]
For each interval $\mathcal{Y}_k$, we repeat the argument of corollary \ref{corr:monotonicPropensityScores}, conditional on $Y_x \in \mathcal{Y}_k$, noting that a nontrivial $\tau$-cdf independence condition will still hold conditional on $Y_x \in \mathcal{Y}_k$.
\end{proof}

\begin{proof}[Proof of corollary \ref{corr:TauIndep_ExtremeValues}]
Let $[a,b] \subseteq [\underline{y}_x,\overline{y}_x] \setminus \mathcal{T}$ with $a < b$. Consider the propensity score
\[
	p(y_x)
	=
	\begin{cases}
		1 &\text{if $y_x \in [a,Q_{Y_x}(F_{Y_x}(a) + \Prob(X=1)\Prob(Y_x \in (a,b))))$} \\
		0 &\text{if $y_x \in [Q_{Y_x}(F_{Y_x}(a) + \Prob(X=1)\Prob(Y_x \in (a,b))),b]$} \\
        \Prob(X=1) &\text{if $y_x \notin [a,b]$}.
	\end{cases}
\]
By assumption A\ref{assn:continuity}.\ref{A1_4}, $\Prob(X=1) \in (0,1)$. This implies that $Q_{Y_x}(F_{Y_x}(a) + \Prob(X=1)\Prob(Y_x \in (a,b))) \in (a,b)$. Therefore, $p$ attains the values 0 and 1 over intervals which have positive Lebesgue measure. Also, by construction we have 
\begin{align*}
	\Exp(p(Y_x) \mid Y_x \in (a,b)) &= \frac{\int_a^b p(y_x) \; dF_{Y_x}(y_x)}{\Prob(Y_x \in (a,b))}\\
	&= \frac{\Prob(Y_x \in [a,Q_{Y_x}(F_{Y_x}(a) + \Prob(X=1)\Prob(Y_x \in (a,b)))))}{\Prob(Y_x \in (a,b))}\\
	&= \frac{F_{Y_x}(a) + \Prob(X=1)\Prob(Y_x \in (a,b)) - F_{Y_x}(a)}{\Prob(Y_x \in (a,b))}\\
	&= \Prob(X = 1).
\end{align*}

Next we show that $\mathcal{T}$-independence holds. Let $t_1$ and $t_2$ be any two values in $\mathcal{T} \cup \{\underline{y}_x, \overline{y}_x \}$ such that $t_1 < t_2$. Then
\[
	\Exp(p(Y_x) \mid Y_x \in (t_1,t_2)) = \Exp(\Prob(X=1) \mid Y_x \in (t_1,t_2)) = \Prob(X=1)
\]
if $t_1 < t_2 < a$ or $b < t_1 < t_2$, by the definition of $p(y_x)$: It equals $\Prob(X=1)$ for values $y_x \notin [a,b]$. Now suppose $t_1<a<b<t_2$. Then
\begin{align*}
	&\Exp(p(Y_x) \mid Y_x \in (t_1,t_2))\\[0.3em]
		&= \Exp(p(Y_x) \mid Y_x \in (t_1,a)) \Prob(Y_x \in (t_1,a) \mid Y_x \in (t_1,t_2)) \\[0.3em]
		&\qquad + \Exp(p(Y_x) \mid Y_x \in (a,b)) \Prob(Y_x \in (a,b) \mid Y_x \in (t_1,t_2)) \\[0.3em]
		&\qquad + \Exp(p(Y_x) \mid Y_x \in (b,t_2)) \Prob(Y_x \in (b,t_2) \mid Y_x \in (t_1,t_2)) \\[0.3em]
		&= \Prob(X=1)\Big(\Prob(Y_x \in (t_1,a) \mid Y_x \in (t_1,t_2)) \\
		&\qquad\qquad\qquad + \Prob(Y_x \in (a,b) \mid Y_x \in (t_1,t_2)) + \Prob(Y_x \in (b,t_2) \mid Y_x \in (t_1,t_2))\Big) \\
		&= \Prob(X=1).
\end{align*}
The first equality follows by iterated expectations. The second equality follows again by the definition of $p(y_x)$ for values $y_x \notin [a,b]$, and also by our derivation above that $\Exp( p(Y_x) \mid Y_x \in (a,b)) = \Prob(X=1)$.

This covers all cases for $t_1$ and $t_2$, thus $\mathcal{T}$-independence holds by theorem \ref{thm:AvgValueCharacterization}.
\end{proof}

\begin{proof}[Proof of proposition \ref{prop:mean-indep prop scores}]
We have
\begin{align*}
	\cov(Y_x,p(Y_x))
		&= \Exp[(Y_x-\Exp(Y_x)) p(Y_x)] \\
		&= \Exp[(Y_x-\Exp(Y_x))(p(Y_x) - p(\Exp(Y_x)))].
\end{align*}
Without loss of generality, suppose $p(y_x)$ is non-decreasing on $\supp(Y_x)$. Therefore $y_x \gtreqless \Exp(Y_x)$ implies $p(y_x) \gtreqless p(\Exp(Y_x))$ and so
\[
	(Y_x-\Exp(Y_x))(p(Y_x) - p(\Exp(Y_x))) \geq 0
\]
with probability one. Moreover, equality holds with probability equal to $\Prob[p(Y_x) = p(\Exp(Y_x))]$. Since $p$ is non-constant and non-decreasing, the probability that $p(Y_x)$ is equal to a constant is strictly less than one. Hence $\cov(Y_x,p(Y_x)) > 0$. 

Next,
\begin{align*}
	\Exp(Y_x \mid X=1) &= \frac{\Exp(Y_xX)}{\Prob(X=1)}\\
	&= \frac{\Exp[Y_x p(Y_x)]}{\Prob(X=1)}\\
	&> \frac{\Exp(Y_x) \Prob(X=1)}{\Prob(X=1)}\\
	&= \Exp(Y_x).
\end{align*}
The first line follows by iterated expectations on $X$. The second line follows by iterated expectations on $Y_x$. The third by $\cov(Y_x,p(Y_x)) > 0$. Thus we have shown that mean independence does not hold.
\end{proof}

\subsection{Proofs for section \ref{sec:treatmentEffects}}

\begin{proof}[Proof of corollary \ref{prop:TauImpliesU}]
If $a = b$ the result holds trivially since $\mathcal{T}$ has measure zero. So suppose $a<b$ and fix $x \in \{0,1\}$. By the properties of conditional probabilities,
\[
	\Prob(Y_0\leq y_0 \mid X=x,Y_0 \in [a,b])
	= \frac{\Prob(Y_0 \in  (\underline{y}_0,y_0] \cap [a,b] \mid X=x)}{\Prob(Y_0\in [a,b] \mid X=x)}.
\]
If $y_0 > b$ this fraction equals 1. If $y_0 < a$ the numerator is zero. In either case, $\Prob(Y_0\leq y_0 \mid X=x,Y_0\in[a,b])$ does not depend on $x$ and hence $Y_0\indep X \mid \{ Y_0 \in \mathcal{T} \}$. 

If $y_0\in[a,b]$, $(\underline{y}_0,y_0] \cap [a,b] = [a,y_0]$. Hence
\begin{align*}
	\Prob(Y_0\leq y_0 \mid X=x,Y_0\in[a,b])
		&= \frac{F_{Y_0 \mid X}(y_0 \mid x) - F_{Y_0 \mid X}(a \mid x)}{F_{Y_0 \mid X}(b \mid x) - F_{Y_0 \mid X}(a \mid x)}\\
		&= \frac{F_{Y_0}(y_0) - F_{Y_0 }(a )}{F_{Y_0 }(b ) - F_{Y_0 }(a )}\\
		&= \Prob(Y_0\leq y_0 \mid Y_0\in[a,b]).
\end{align*}
The first line follows since $Y_0 \mid X=x$ is continuously distributed, by lemma \ref{lemma:continuity}. The second line follows from $a,y_0,b\in\mathcal{T}$. Hence $Y_0\indep X \mid \{ Y_0 \in \mathcal{T} \}$. This implies that, for almost all $y_0 \in \mathcal{T}$,
\begin{align*}
	\Prob(X=1 \mid Y_0=y_0)
		&= \Prob(X=1 \mid Y_0=y_0,Y_0\in\mathcal{T}) \\
		&= \Prob(X=1 \mid Y_0\in\mathcal{T}).
\end{align*}
The second equality follows from conditional independence. Thus we have shown that $p(y_0)$ equals a constant on $\mathcal{T}$. To finish the proof, we show that this constant is $\Prob(X=1)$. Let $t_1, t_2 \in [a,b]$ with $t_1 < t_2$. Then
\begin{align*}
	\Prob(X=1)
		&= \Exp(p(Y_0) \mid Y_0 \in (t_1,t_2)) \\
		&= \Exp(\Prob(X=1 \mid Y_0\in\mathcal{T}) \mid Y_0 \in (t_1,t_2))\\
		&= \Prob(X=1 \mid Y_0 \in \mathcal{T}).
\end{align*}
The first line follows by theorem \ref{thm:AvgValueCharacterization} while the second line follows by our derivations above showing that $p(y_0) = \Prob(X=1 \mid Y_0\in\mathcal{T})$ on $\mathcal{T}$. Hence
\[
	\Prob(X=1 \mid Y_0=y_0) = \Prob(X=1)
\]
for almost all $y_0 \in \mathcal{T}$.
\end{proof}

\subsubsection*{Preliminaries for the Proof of Proposition \ref{prop:quantileTrtBounds}}

To obtain the identified sets in section \ref{sec:treatmentEffectBounds}, we first derive sharp bounds on cdfs under generic $\mathcal{T}$- and $\mathcal{U}$-independence. We then apply these results to obtain sharp bounds on the treatment effect parameters. To this end, in this subsection we consider the relationship between a generic continuous random variable $U$ and a binary variable $X \in \{ 0,1 \}$. We derive sharp bounds on the conditional cdf of $U$ given $X$ when (1) the marginal distributions of $U$ and $X$ are known and either (2) $U$ is $\mathcal{T}$-independent of $X$ or (2$^\prime$) $U$ is $\mathcal{U}$-independent of $X$. We obtain the identified sets in section \ref{sec:treatmentEffectBounds} by applying this general result to $U = F_{Y_x}(Y_x)$.%

Let $F_{U \mid X}(u \mid x) = \Prob(U \leq u \mid X=x)$ denote the unknown conditional cdf of $U$ given $X=x$. Let $F_U(u) = \Prob(U \leq u)$ denote the known marginal cdf of $U$. Let $p_x = \Prob(X=x)$ denote the known marginal probability mass function of $X$. Let $a,b \in \R$, $a \leq b$. We define the functions $\overline{F}_{U \mid X}^\mathcal{T}(u \mid x)$, $\underline{F}_{U \mid X}^\mathcal{T}(u \mid x)$, $\overline{F}_{U \mid X}^\mathcal{U}(u \mid x)$, and $\underline{F}_{U \mid X}^\mathcal{U}(u \mid x)$ in appendix \ref{appendix:cdfbounds}. These are piecewise linear functions of $F_U(u)$ which depend on $a$, $b$, and $p_x$. Figure \ref{TauDepBoundsOnCDFs} plots several examples.

\begin{proposition}\label{prop:TauCDFbounds}
Suppose the following hold:
\begin{enumerate}
\item The marginal distributions of $U$ and $X$ are known.
\item $U$ is continuously distributed.
\item $p_1 \in (0,1)$.
\item $U$ is $\mathcal{T}$-independent of $X$ with $\mathcal{T} = [a,b]$.
\end{enumerate}
Let $\mathcal{F}_{\supp(U)}$ denote the set of all cdfs on $\supp(U)$. Then, for each $x \in \{0,1\}$, $F_{U \mid X}(\cdot \mid x) \in \mathcal{F}_{U \mid x}^{\mathcal{T}}$, where
\[
	\mathcal{F}_{U \mid x}^\mathcal{T} =
	\left\{ F \in \mathcal{F}_{\supp(U)} :
	\underline{F}_{U \mid X}^\mathcal{T}(u \mid x) \leq F_{U \mid X}(u \mid x) \leq \overline{F}_{U \mid X}^\mathcal{T}(u \mid x) \text{ for all $u \in \supp(U)$} \right\}.
\]
Furthermore, for each $\varepsilon \in [0,1]$, there exists a joint distribution of $(U,X)$ consistent with assumptions 1--4 above and such that
\begin{multline}\label{eq:joint epsilon cdf}
	\big( \Prob(U\leq u \mid X=1), \, \Prob(U\leq u \mid X=0) \big) \\
	= \left( \varepsilon  \underline{F}_{U \mid X}^\mathcal{T}(u \mid 1) + (1-\varepsilon) \overline{F}_{U \mid X}^{\mathcal{T}}(u \mid 1), \, 
		(1-\varepsilon)  \underline{F}_{U \mid X}^\mathcal{T}(u \mid 0) + \varepsilon \overline{F}_{U \mid X}^{\mathcal{T}}(u \mid 0) \right)
\end{multline}
for all $u \in \supp(U)$. Finally, the entire proposition continues to hold if we replace $\mathcal{T}$ with $\mathcal{U}$.
\hfill $\blacksquare$
\end{proposition}

\bigskip

Consider the $\mathcal{T}$-independence case. Then proposition \ref{prop:TauCDFbounds} has two conclusions. First, we show that the functions $\overline{F}_{U \mid X}^\mathcal{T}(\cdot \mid x)$ and $\underline{F}_{U \mid X}^\mathcal{T}(\cdot \mid x)$ bound the unknown conditional cdf $F_{U \mid X}(\cdot \mid x)$ uniformly in its arguments. Second, we show that these bounds are functionally sharp in the sense that the joint identified set for the two conditional cdfs $(F_{U \mid X}(\cdot \mid 1), F_{U \mid X}(\cdot \mid 0))$ contains linear combinations of the bound functions $\overline{F}_{U \mid X}^\mathcal{T}(\cdot \mid x)$ and $\underline{F}_{U \mid X}^\mathcal{T}(\cdot \mid x)$. We use this second conclusion to prove sharpness of our treatment effect parameters of section \ref{sec:treatmentEffectBounds}. Identical conclusions hold in the $\mathcal{U}$-independence case.

For simplicity we have only stated this result when $\mathcal{T}$ and $\mathcal{U}$ are closed intervals $[a,b]$. It can be generalized, however. For example, for $\mathcal{T}$-independence, theorem 2 of \cite{MastenPoirier2016} provides cdf bounds when $\mathcal{T}$ is a finite union of closed intervals. 

\begin{figure}[t]
\centering
\includegraphics[width=50mm]{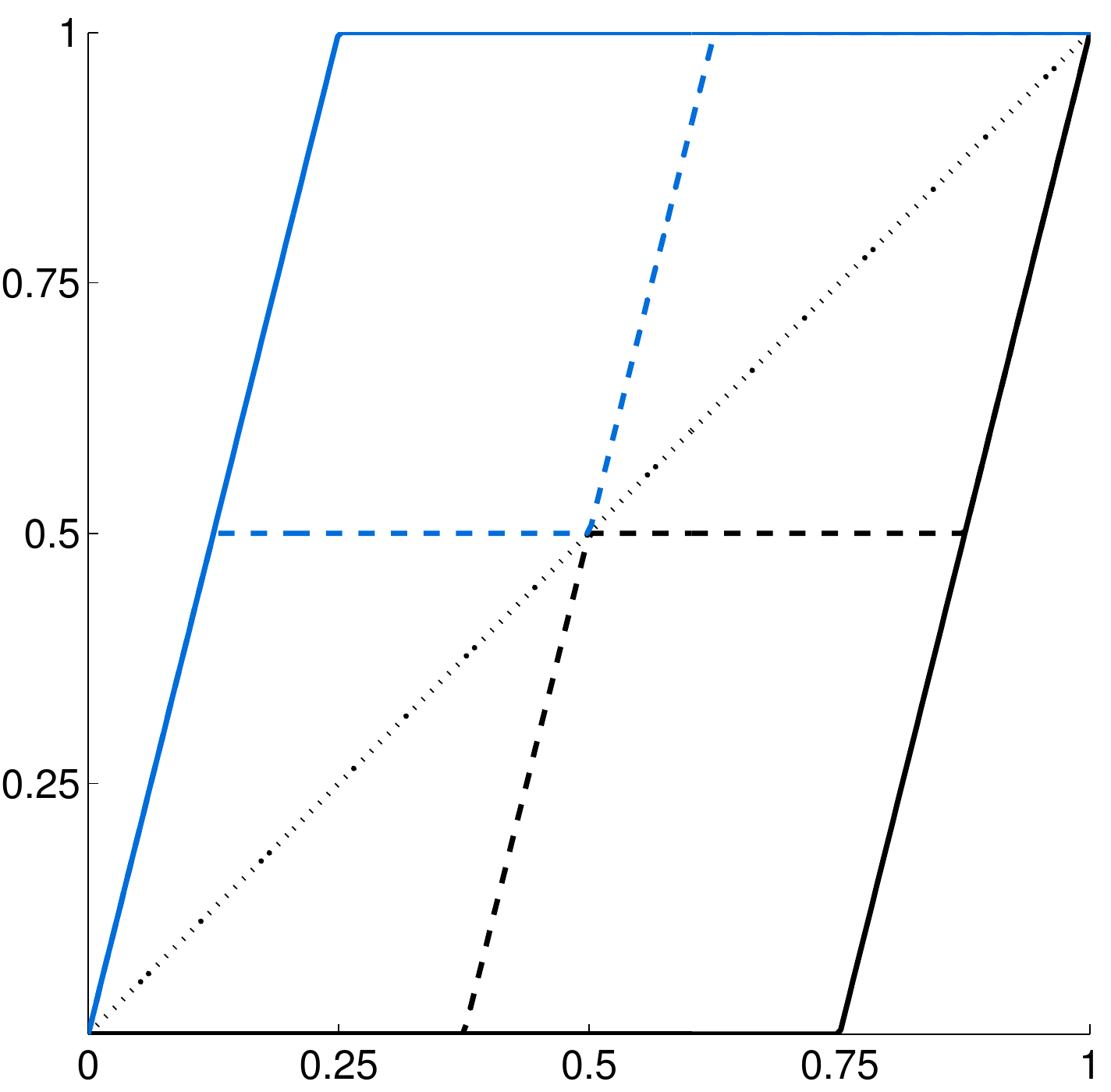}
\hspace{3mm}
\includegraphics[width=50mm]{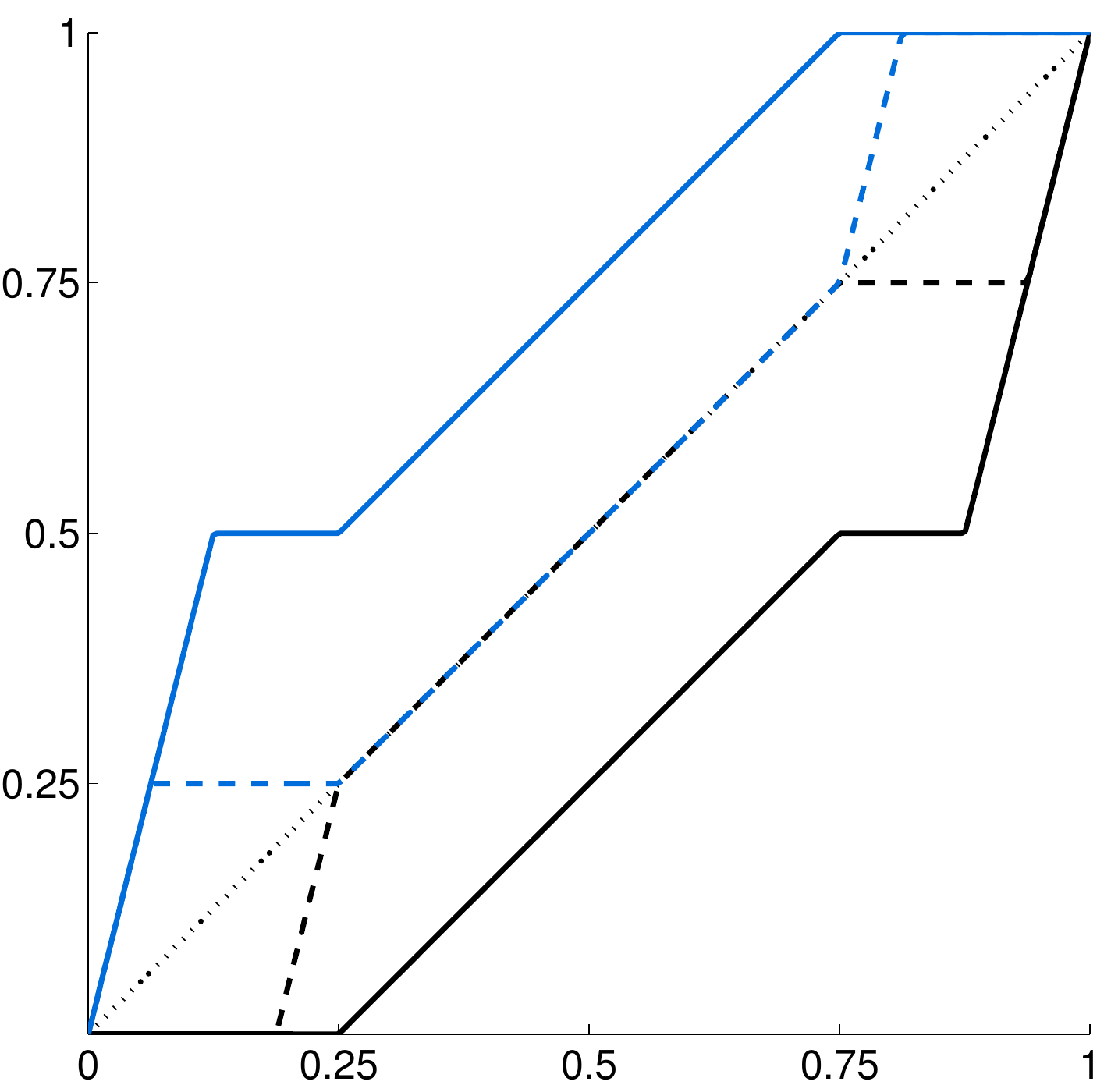}
\hspace{3mm}
\includegraphics[width=50mm]{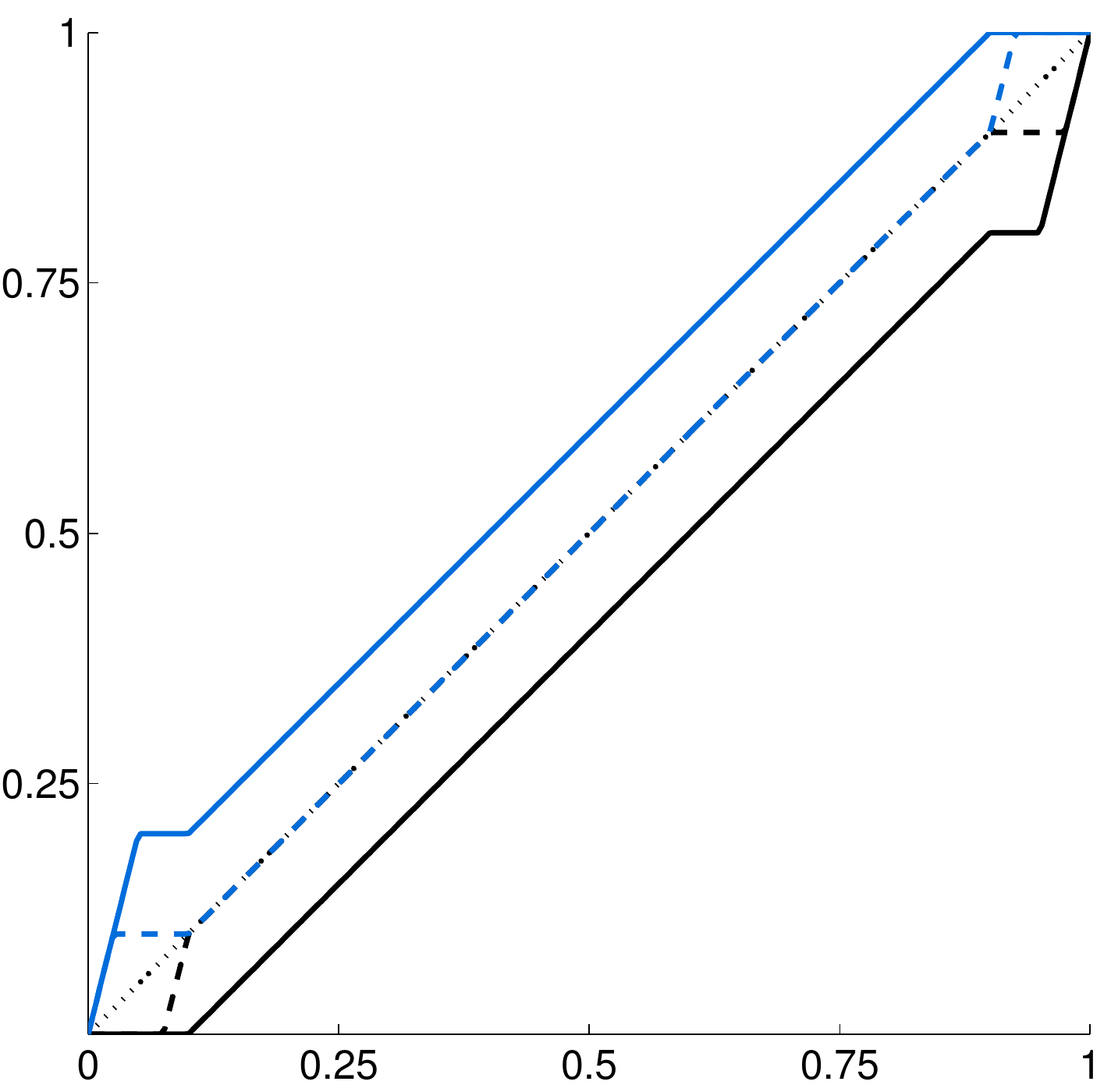}
\caption{Example upper and lower bounds on $F_{U \mid X}(u \mid 0)$, when $p_1 = 0.75$ and $U \sim \text{Unif}[0,1]$. Solid: $\mathcal{U}$-independence. Dashed: $\mathcal{T}$-independence. All three plots have $\mathcal{U} = \mathcal{T}$, for three different choices: $\{ 0.5 \}$ on the left, $[0.25, 0.75]$ in the middle, and $[0.1,0.9]$ on the right. The diagonal, representing the choice $[0,1]$---the case of full independence---is plotted as a dotted line.}
\label{TauDepBoundsOnCDFs}
\end{figure}

We now begin by proving proposition \ref{prop:TauCDFbounds}. To do so, we frequently use the following result.

\begin{lemma}\label{lemma:cdfAsIntegralOfPropensityScore}
Let $U$ be a continuous random variable. Let $X$ be a random variable with $p_x = \Prob(X=x) > 0$. Then
\[
	F_{U \mid X}(u \mid x) = \int_{-\infty}^u \frac{\Prob(X=x \mid U=v)}{p_x} \; dF_U(v).
\]
\end{lemma}

\begin{proof}[Proof of lemma \ref{lemma:cdfAsIntegralOfPropensityScore}]
See lemma 1 in \cite{MastenPoirier2017}.
\end{proof}

\begin{proof}[Proof of proposition \ref{prop:TauCDFbounds} ($\mathcal{T}$-independence)]
We prove this statement for $\mathcal{T}$-independence first, then for $\mathcal{U}$-independence. Both proofs proceed by first deriving the upper cdf bound, then deriving the lower cdf bound, and finishing by showing the joint attainability of the cdfs of equation \eqref{eq:joint epsilon cdf}.

\bigskip

For both the $\mathcal{T}$- and $\mathcal{U}$-independence proofs, we use the following two inequalities: First, for all $u \in \supp(U)$,
\begin{align}\label{eq:uniformUpperBound}
	F_{U \mid X}(u \mid x)
		&= \int_{-\infty}^u \frac{\Prob(X=x \mid U=v)}{p_x} \; dF_U(v) \notag\\
		&\leq \int_{-\infty}^u \frac{1}{p_x} \; dF_U(v) \notag\\
		&= \frac{F_U(u)}{p_x}.
\end{align}
The first line follows by lemma \ref{lemma:cdfAsIntegralOfPropensityScore}. The second line follows by $\Prob(X=x \mid U=v) \leq 1$. Second, for all $u \in \supp(U)$,
\begin{align}\label{eq:uniformLowerBound}
	F_{U \mid X}(u \mid x)
		&= \int_{-\infty}^u \frac{\Prob(X=x \mid U=v)}{p_x} \; dF_U(v) \notag \\
		&= 1 - \int_u^\infty \frac{\Prob(X=x \mid U=v)}{p_x} \; dF_U(v) \notag \\
		&\geq 1 - \int_u^\infty \frac{1}{p_x} \; dF_U(v) \notag \\
		&= 1 + \frac{F_U(u) - 1}{p_x}.
\end{align}
While equations \eqref{eq:uniformUpperBound} and \eqref{eq:uniformLowerBound} both hold for all $u \in \supp(U)$, they are not sharp for all $u$.

\bigskip

\textbf{Part 1.} We show that $F_{U \mid X}(u \mid x) \leq \overline{F}_{U \mid X}^{\mathcal{T}}(u \mid x)$ for all $u \in \supp(U)$. If $u \leq Q_U(p_xF_U(a))$, the upper bound holds by equation \eqref{eq:uniformUpperBound}.
Second, if $u \in [Q_U(p_xF_U(a)),a]$, then $F_{U \mid X}(u \mid x) \leq F_{U \mid X}(a \mid x) = F_U(a)$ since $F_{U \mid X}(\cdot \mid x)$ is nondecreasing and by $\mathcal{T}$-independence. Third, if $u\in [a,b]$, then $F_{U \mid X}(u \mid x) = F_U(u)$ by $\mathcal{T}$-independence. Fourth, if $u \in [b, Q_U(p_x + F_U(b)(1-p_x))]$, then 
\begin{align*}
	F_{U \mid X}(u \mid x)
		&= \int_{-\infty}^u \frac{\Prob(X=x \mid U=v)}{p_x} \; dF_U(v)\\
		&= \int_{-\infty}^b \frac{\Prob(X=x \mid U=v)}{p_x} \; dF_U(v)
			+ \int_b^u \frac{\Prob(X=x \mid U=v)}{p_x} \; dF_U(v)\\
		&\leq F_{U \mid X}(b \mid x) + \int_b^u \frac{1}{p_x} \; dF_U(v)\\
		&= F_U(b) + \frac{F_U(u) - F_U(b)}{p_x}.
\end{align*}

Finally, for all $u \in \supp(U)$, $F_{U \mid X}(u \mid x) \leq 1$. In particular, this holds for $u \geq Q_U(p_x + F_U(b)(1-p_x))$.

\bigskip

\textbf{Part 2.} We show that $F_{U \mid X}(u \mid x) \geq \underline{F}_{U \mid X}^{\mathcal{T}}(u \mid x)$ for all $u \in \supp(U)$. First, $F_{U \mid X}(u \mid x) \geq 0$ for all $u \in \supp(U)$. In particular, this holds for $u \leq Q_{U}((1-p_x)F_U(a))$. Second, if $u \in [Q_{U}((1-p_x)F_U(a)),a]$, then 
\begin{align*}
	F_{U \mid X}(u \mid x)
		&= \int_{-\infty}^u \frac{\Prob(X=x \mid U=v)}{p_x} \; dF_U(v)\\
		&= \int_{-\infty}^a \frac{\Prob(X=x \mid U=v)}{p_x} \; dF_U(v)
			- \int_u^a \frac{\Prob(X=x \mid U=v)}{p_x} \; dF_U(v)\\
		&\geq F_{U \mid X}(a \mid x) - \int_u^a \frac{1}{p_x} \; dF_U(v)\\
		&= F_U(a) + \frac{F_U(u) - F_U(a)}{p_x}.
\end{align*}
Third, if $u\in [a,b]$ then $F_{U \mid X}(u \mid x) = F_U(u)$ by $\mathcal{T}$-independence. Fourth, if $u\in[b,Q_U(p_xF_U(b) + (1-p_x))]$, then $F_{U \mid X}(u \mid x) \geq F_{U \mid X}(b \mid x) = F_{U}(b)$. Finally, if $u \geq Q_U(p_xF_U(b) + (1-p_x))$, the lower bound holds by equation \eqref{eq:uniformLowerBound}.

\bigskip

\textbf{Part 3}. To prove sharpness, we must construct a joint distribution of $(U,X)$ consistent with assumptions 1--4 and which yields the upper bound $\overline{F}_{U \mid X}^\mathcal{T}(\cdot \mid x)$. And likewise for the lower bound $\underline{F}_{U \mid X}^\mathcal{T}(\cdot \mid x)$. This yields equation \eqref{eq:joint epsilon cdf} for $\varepsilon = 0$ and $\varepsilon = 1$. By taking convex combinations of these two joint distributions we obtain the case for $\varepsilon \in (0,1)$.

The marginal distributions of $U$ and $X$ are prespecified. Hence to construct the joint distribution of $(U,X)$ it suffices to define conditional distributions of $U \mid X$. We define these conditional distributions by the bound functions themselves, $\underline{F}_{U \mid X}^{\mathcal{T}}(u \mid x)$ and $\overline{F}_{U \mid X}^{\mathcal{T}}(u \mid x)$. These functions are non-decreasing, right-continuous, and have range $[0,1]$. Hence they are valid cdfs. They also satisfy $\mathcal{T}$-independence. These properties are preserved by taking convex combinations, and hence $\varepsilon \underline{F}_{U \mid X}^{\mathcal{T}}(u \mid x) + (1-\varepsilon)\overline{F}_{U \mid X}^{\mathcal{T}}(u \mid x)$ is also a valid cdf that satisfies $\mathcal{T}$-independence for any $\varepsilon \in [0,1]$ and $x\in\{0,1\}$. Finally, we show that these cdfs are consistent with the marginal distribution of $U$, and can satisfy both components of equation \eqref{eq:joint epsilon cdf} simultaneously. To see this, we compute
\begin{align*}
	&p_x \overline{F}_{U \mid X}^\mathcal{T}(u \mid x) + (1-p_x) \underline{F}_{U \mid X}^\mathcal{T}(u \mid 1-x) \\
	&=
	\begin{cases}
	p_x\dfrac{F_U(u)}{p_x}
		&\text{if $u\leq Q_{U}(p_xF_U(a))$}\\
	p_xF_U(a) + F_U(u) - F_U(a) + F_U(a)(1-p_x)
		&\text{if $Q_{U}(p_xF_U(a)) \leq u \leq a$}\\
	p_xF_U(u) + (1-p_x)F_U(u)
		&\text{if $a\leq u \leq b$} \\
	F_U(u) - F_U(b) + p_xF_U(b) + (1-p_x) F_U(b)
		&\text{if $b\leq u \leq Q_{U}(p_x + F_U(b)(1-p_x))$}\\	p_x + F_U(u) - 1 + (1-p_x)
		&\text{if $Q_{U}(p_x + F_U(b)(1-p_x))\leq u$}
	\end{cases} \\
	&= F_U(u).
\end{align*}
Thus
\begin{align*}%
	&p_1 \left[ \varepsilon \underline{F}_{U \mid X}^{\mathcal{T}}(u \mid 1) 
		+ (1-\varepsilon)\overline{F}_{U \mid X}^{\mathcal{T}}(u \mid 1) \right] 
	+ p_0 \left[ (1-\varepsilon)\underline{F}_{U \mid X}^{\mathcal{T}}(u \mid 0)
		+ \varepsilon \overline{F}_{U \mid X}^{\mathcal{T}}(u \mid 0) \right] \\%
	&= \varepsilon \left[ p_1 \underline{F}_{U \mid X}^\mathcal{T}(u \mid 1)
		+ p_0 \overline{F}_{U \mid X}^\mathcal{T}(u \mid 0) \right] 
		+ (1-\varepsilon) \left[ p_1 \overline{F}_{U \mid X}^\mathcal{T}(u \mid 1)
		+ p_0 \underline{F}_{U \mid X}^\mathcal{T}(u \mid 0) \right] \\%
	&= \varepsilon F_U(u) + (1-\varepsilon) F_U(u) \\ %
	&= F_U(u).
\end{align*}
\end{proof}

\begin{proof}[Proof of proposition \ref{prop:TauCDFbounds} ($\mathcal{U}$-independence)]
Now we consider the cdf bounds under $\mathcal{U}$-independence, under various cases: 

\bigskip

\textbf{Part 1.} We show $F_{U \mid X}(u \mid x) \leq \overline{F}_{U \mid X}^{\mathcal{U}}(u \mid x)$ for all $u \in \supp(U)$. We do this in two cases.

\medskip

\textbf{Part 1a.} Suppose $(1 - (F_U(b) - F_U(a)))p_x \leq F_U(a)$. First, $F_{U \mid X}(u \mid x) \leq 1$ for all $u \in \supp(U)$. In particular, this holds if $u \geq b$.

Second, if $u \in [a,b]$, then
\begin{align*}
	F_{U \mid X} (u \mid x)
		&= \int_{-\infty}^u \frac{\Prob(X=x \mid U=v)}{p_x} \; dF_U(v) \\
		&= 1- \int_{b}^\infty \frac{\Prob(X=x \mid U=v)}{p_x} \; dF_U(v)
			- \int_{u}^b \frac{\Prob(X=x \mid U=v)}{p_x} \; dF_U(v) \\
		&\leq 1 - \int_{b}^\infty \frac{0}{p_x} \; dF_U(v)
			- \int_{u}^b \frac{p_x}{p_x} \; dF_U(v) \\
		&= 1 - (F_U(b) - F_U(u)).
\end{align*}
Third, if $u \in [Q_U((1-(F_U(b) - F_U(a)))p_x), a]$, then $F_{U \mid X}(u \mid x) \leq F_{U \mid X}(a \mid x) \leq 1 - (F_U(b) - F_U(a))$ where the last inequality follows by our derivation immediately above. Finally, if $u \leq Q_{U}((1-(F_U(b) - F_U(a)))p_x)$, the upper bound holds by equation \eqref{eq:uniformUpperBound}.

\bigskip

\textbf{Part 1b.} Now suppose $(1 - (F_U(b) - F_U(a)))p_x \geq F_U(a)$. First, if $u \leq a$, the upper bound holds by equation \eqref{eq:uniformUpperBound}. Second, if $u \in [a,b]$ then 
\begin{align*}
	F_{U \mid X} (u \mid x)
		&= \int_{-\infty}^u \frac{\Prob(X=x \mid U=v)}{p_x} \; dF_U(v) \\
		&= \int_{-\infty}^a \frac{\Prob(X=x \mid U=v)}{p_x} \; dF_U(v)
			+ \int_{a}^u \frac{\Prob(X=x \mid U=v)}{p_x} \; dF_U(v)\\
		&\leq \int_{-\infty}^a \frac{1}{p_x} \; dF_U(v)
			+ \int_{a}^u \frac{p_x}{p_x} \; dF_U(v)\\
		&= \frac{F_{U}(a)}{p_x} + F_U(u) - F_U(a).
\end{align*}
Third, if $u \in [b,Q_U((F_U(b)-F_U(a))(1-p_x)+ p_x)]$, then 
\begin{align*}
	&F_{U \mid X} (u \mid x) \\
		&= \int_{-\infty}^u \frac{\Prob(X=x \mid U=v)}{p_x} \; dF_U(v) \\
		&= \int_{-\infty}^a \frac{\Prob(X=x \mid U=v)}{p_x} \; dF_U(v)
			+ \int_{a}^b \frac{\Prob(X=x \mid U=v)}{p_x} \; dF_U(v)
			+ \int_{b}^u \frac{\Prob(X=x \mid U=v)}{p_x} \; dF_U(v) \\
		&\leq \int_{-\infty}^a \frac{1}{p_x} \; dF_U(v)
			+ \int_{a}^b \frac{p_x}{p_x} \; dF_U(v)
			+ \int_{b}^u \frac{1}{p_x} \; dF_U(v) \\
		&= \frac{F_U(a) + p_x(F_U(b) - F_U(a)) + F_U(u) - F_U(b)}{p_x}.
\end{align*}
Finally, if $u \geq Q_U((F_U(b)-F_U(a))(1-p_x)+ p_x)$, then $F_{U \mid X}(u \mid x) \leq 1$. 

\bigskip

\textbf{Part 2.} We show that $F_{U \mid X}(u \mid x) \geq \underline{F}_{U \mid X}^{\mathcal{U}}(u \mid x)$ for all $u \in \supp(U)$. We do this in two cases.

\medskip

\textbf{Part 2a.} Suppose $(1 - (F_U(b) - F_U(a)))(1-p_x) \leq F_U(a)$. First, if $u \leq Q_U((1-(F_U(b) - F_U(a)))(1-p_x))$, then $F_{U \mid X}(u \mid x) \geq 0$.
Second, if $u \in [Q_U((1-(F_U(b) - F_U(a)))(1-p_x)),a]$, then
\begin{align*}
	&F_{U \mid X} (u \mid x) \\
		&= \int_{-\infty}^u \frac{\Prob(X=x \mid U=v)}{p_x} \; dF_U(v) \\
		&= 1- \int_{b}^\infty \frac{\Prob(X=x \mid U=v)}{p_x} \; dF_U(v)
			- \int_{a}^b \frac{\Prob(X=x \mid U=v)}{p_x} \; dF_U(v)
			- \int_{u}^a \frac{\Prob(X=x \mid U=v)}{p_x} \; dF_U(v) \\
		&\geq 1- \int_{b}^\infty \frac{1}{p_x} \; dF_U(v)
			- \int_{a}^b \frac{p_x}{p_x} \; dF_U(v)
			- \int_{u}^a \frac{1}{p_x} \; dF_U(v) \\
		&=  \frac{F_U(u) - (1 - (F_U(b) -F_U(a)))(1-p_x) }{p_x}.
\end{align*}

Third, if $u\in [a,b]$, then 
\begin{align*}
	F_{U \mid X} (u \mid x)
		&= \int_{-\infty}^u \frac{\Prob(X=x \mid U=v)}{p_x} \; dF_U(v) \\
		&= 1- \int_{b}^\infty \frac{\Prob(X=x \mid U=v)}{p_x} \; dF_U(v)
			- \int_{u}^b \frac{\Prob(X=x \mid U=v)}{p_x} \; dF_U(v) \\
		&\geq 1 - \int_{b}^\infty \frac{1}{p_x} \; dF_U(v)
			- \int_{u}^b \frac{p_x}{p_x} \; dF_U(v) \\
		&= F_U(u) + \frac{(F_U(b) - 1)(1-p_x)}{p_x}.
\end{align*}
Finally, if $u \geq b$, the lower bound holds by equation \eqref{eq:uniformLowerBound}.

\bigskip

\textbf{Part 2b.} Now suppose $(1 - (F_U(b) - F_U(a)))(1-p_x) \geq F_U(a)$. First, if $u\leq a$ then $F_{U \mid X}(u \mid x) \geq 0$. Second, if $u\in [a,b]$ then 
\begin{align*}
	F_{U \mid X} (u \mid x)
		&= \int_{-\infty}^u \frac{\Prob(X=x \mid U=v)}{p_x} \; dF_U(v) \\
		&= \int_{-\infty}^a \frac{\Prob(X=x \mid U=v)}{p_x} \; dF_U(v)
			+ \int_{a}^u \frac{\Prob(X=x \mid U=v)}{p_x} \; dF_U(v) \\
		&\geq \int_{-\infty}^a \frac{0}{p_x} \; dF_U(v)
			+ \int_{a}^u \frac{p_x}{p_x} \; dF_U(v) \\
		&= F_U(u) - F_U(a).
\end{align*}
Third, if $u\in [b,Q_U(p_x(F_U(b)-F_U(a)) + 1-p_x)]$, then $F_{U \mid X}(u \mid x) \geq F_{U \mid X}(b \mid x) \geq F_{U}(b) - F_U(a)$, where the last inequality follows by our derivation immediately above. Finally, if $u \geq Q_U(p_x(F_U(b)-F_U(a)) + 1-p_x)$ the lower bound holds by equation \eqref{eq:uniformLowerBound}.

\bigskip

\textbf{Part 3.} In this part, we prove sharpness in two steps. First we construct a joint distribution of $(U,X)$ consistent with assumptions 1--4 and which yields the upper bound $\overline{F}_{U \mid X}^\mathcal{U}(\cdot \mid x)$. And likewise for the lower bound $\underline{F}_{U \mid X}^\mathcal{U}(\cdot \mid x)$. This yields equation \eqref{eq:joint epsilon cdf} for $\varepsilon = 0$ and $\varepsilon = 1$. Second we use convex combinations of these two joint distributions to obtain the case for $\varepsilon \in (0,1)$.

The marginal distributions of $U$ and $X$ are prespecified. Hence to construct the joint distribution of $(U,X)$ it suffices to define conditional distributions of $X \mid U$. Specifically, when $(1 - (F_U(b) - F_U(a)))p_x \leq F_U(a)$, define the conditional probability
\[
	\overline{p}_x(u) =
	\begin{cases}
		1 &\text{ for $u < Q_U((1-(F_U(b)-F_U(a)))p_x)$} \\
		0 &\text{ for $u \in [(Q_U((1-(F_U(b)-F_U(a)))p_x),a)$} \\
		p_x &\text{ for $u \in [a,b)$} \\
		0 &\text{ for $u \geq b$.}
	\end{cases}
\]
for $u\in\supp(U)$. This conditional probability is consistent with $\mathcal{U}$-independence. Moreover, by applying lemma \ref{lemma:cdfAsIntegralOfPropensityScore} one can verify that it yields the upper bound $\overline{F}_{U \mid X}^{\mathcal{U}}(\cdot \mid x)$.

When $(1 - (F_U(b) - F_U(a)))p_x \geq F_U(a)$, define
\[
	\overline{p}_x(u) =
	\begin{cases}
		1 &\text{ for $u < a$} \\
		p_x &\text{ for $u \in [a,b)$} \\
		1 &\text{ for $u \in [b,Q_U((F_U(b)-F_U(a))(1-p_x)+ p_x))$} \\
		0 &\text{ for $u \geq Q_U((F_U(b)-F_U(a))(1-p_x)+ p_x)$.}
	\end{cases}
\]
Again, by applying lemma \ref{lemma:cdfAsIntegralOfPropensityScore} one can verify that this conditional probability yields the upper bound $\overline{F}_{U \mid X}^{\mathcal{U}}(\cdot \mid x)$.

Next consider the lower bounds. When $(1 - (F_U(b) - F_U(a)))(1-p_x) \leq F_U(a)$, define
\[
	\underline{p}_x(u) =
	\begin{cases}
		0 &\text{ for } u < Q_U((1-(F_U(b)-F_U(a)))(1-p_x))\\
		1 &\text{ for } u\in[Q_U((1-(F_U(b)-F_U(a)))(1-p_x)),a)\\
		p_x &\text{ for } u\in[a,b)\\
		1 &\text{ for } u\geq b.
	\end{cases}
\]
When $(1 - (F_U(b) - F_U(a)))(1-p_x) \geq F_U(a)$, define
\[
	\underline{p}_x(u) =
	\begin{cases}
		0 &\text{ for } u < a\\
		p_x &\text{ for } u\in[a,b)\\
		0 &\text{ for } u\in[b, Q_U(p_x(F_U(b) - F_U(a)) + 1-p_x))\\
		1 &\text{ for } u\geq Q_U(p_x(F_U(b) - F_U(a)) + 1-p_x).
	\end{cases}
\]
As with the upper bounds, one can verify that these yield the lower bound $\underline{F}_{U \mid X}^{\mathcal{U}}(\cdot \mid x)$. For all of these conditional distributions of $X \mid U$, one can verify that they are consistent with the marginal distribution of $X$:
\[
	\int_{\supp(U)} \overline{p}_x(u) \; dF_U(u) = p_x
	\qquad \text{and} \qquad
	\int_{\supp(U)} \underline{p}_x(u) \; dF_U(u) = p_x.
\]
Thus we have shown that the bound functions are attainable. That is, equation \eqref{eq:joint epsilon cdf} holds with $\varepsilon = 0$ or $1$. Next consider $\varepsilon \in (0,1)$. For this $\varepsilon$, we specify the distribution of $X \mid U$ by the conditional probability $\varepsilon \underline{p}_x(u) + (1-\varepsilon) \overline{p}_x(u)$. This is a valid conditional probability since it is a convex combination of two terms which are between 0 and 1. This conditional probability satisfies $\mathcal{U}$-independence. By linearity of integrals and our results above,
\[
	\int_{\supp(U)} \left[ \varepsilon \underline{p}_x(u) + (1-\varepsilon) \overline{p}_x(u) \right] \; dF_U(u) = p_x
\]
and hence this distribution of $X \mid U$ is consistent with the marginal distribution of $X$. Finally, by lemma \ref{lemma:cdfAsIntegralOfPropensityScore} and linearity of the integral, this conditional probability yields the cdf
\[
	\Prob(U \leq u \mid X=x) = \varepsilon \underline{F}_{U \mid X}^{\mathcal{U}}(u \mid x) + (1-\varepsilon)\overline{F}_{U \mid X}^{\mathcal{U}}(u \mid x),
\]
as needed for each component of equation \eqref{eq:joint epsilon cdf}. To see that each component of equation \eqref{eq:joint epsilon cdf} holds simultaneously, we show that a law of total probability constraint holds. There are two cases to check. First suppose $(1 - (F_U(b) - F_U(a)))(1-p_1)  = (1 - (F_U(b) - F_U(a)))p_0 \leq F_U(a)$. Then
\begin{align*}
	&p_1 \underline{F}_{U \mid X}^\mathcal{U}(u \mid 1) + 	p_0 \overline{F}_{U \mid X}^\mathcal{U}(u \mid 0)\\
 	&=
	\begin{cases}
		0 + F_U(u) \\
		\qquad \text{ for $u < Q_U((1-(F_U(b)-F_U(a)))p_0)$} \\
		\\
		(F_U(u) - (1-(F_U(b) - F_U(a))))p_0 + (1-(F_U(b) - F_U(a)))p_0 \\
		\qquad \text{ for $u \in [Q_U((1-(F_U(b)-F_U(a)))p_0),a)$} \\
		\\
		(F_U(b)-1)p_0 + F_U(u)p_1 + (1 - (F_U(b) - F_U(u))p_0
		&\text{ for $u \in [a,b)$} \\
		F_U(u) - 1 + p_1 + p_0
			&\text{ for $u \geq b$}
	\end{cases}\\
	&= F_U(u).
\end{align*}
Likewise, $p_1 \overline{F}_{U \mid X}^\mathcal{U}(u \mid 1) + 	p_0 \underline{F}_{U \mid X}^\mathcal{U}(u \mid 0) = F_U(u)$. Similar derivations hold for the other case. Thus
\begin{align*}
	&p_1 \left[ \varepsilon \underline{F}_{U \mid X}^{\mathcal{U}}(u \mid 1) 
		+ (1-\varepsilon)\overline{F}_{U \mid X}^{\mathcal{U}}(u \mid 1) \right] 
	+ p_0 \left[ (1-\varepsilon)\underline{F}_{U \mid X}^{\mathcal{U}}(u \mid 0)
		+ \varepsilon \overline{F}_{U \mid X}^{\mathcal{U}}(u \mid 0) \right] \\
	&= \varepsilon \left[ p_1 \underline{F}_{U \mid X}^\mathcal{U}(u \mid 1)
		+ p_0 \overline{F}_{U \mid X}^\mathcal{U}(u \mid 0) \right] 
		+ (1-\varepsilon) \left[ p_1 \overline{F}_{U \mid X}^\mathcal{U}(u \mid 1)
		+ p_0 \underline{F}_{U \mid X}^\mathcal{U}(u \mid 0) \right] \\
	&= \varepsilon F_U(u) + (1-\varepsilon) F_U(u) \\
	&= F_U(u).
\end{align*}
\end{proof}

Using proposition \ref{prop:TauCDFbounds}, we can now prove proposition \ref{prop:quantileTrtBounds}.

\begin{proof}[Proof of proposition \ref{prop:quantileTrtBounds}]
By the law of total probability and equation \eqref{eq:potential outcomes},
\begin{align*}
	F_{Y_0}(y)
		&= F_{Y_0 \mid X}(y \mid 1) p_1 + F_{Y_0 \mid X}(y \mid 0) p_0 \\
		&= F_{Y_0 \mid X}(y \mid 1) p_1 + F_{Y \mid X}(y \mid 0) p_0.
\end{align*}
Rearranging yields
\begin{equation}\label{eq:cdfY0givenX}
	F_{Y_0 \mid X}(y \mid 1) = \frac{F_{Y_0}(y) - p_0 F_{Y \mid X}(y \mid 0)}{p_1}.
\end{equation}
Finally, the desired quantile is simply the left-inverse of this conditional cdf:
\[
	Q_{Y_0 \mid X}(\tau \mid 1) = F_{Y_0 \mid X}^{-1}(\tau \mid 1).
\]
Thus it suffices to obtain bounds on the unconditional cdf $F_{Y_0}(y)$. By equation (12) on page 337 of \cite{MastenPoirier2017},
\begin{equation}\label{eq:eq12MP2018}
	F_{Y \mid X}(y \mid 0) = F_{R_0 \mid X}(F_{Y_0}(y) \mid 0)
\end{equation}
where $R_0 \equiv F_{Y_0}(Y_0)$ is the \emph{rank} of $Y_0$. By A\ref{assn:continuity}, $Y_0$ is continuously distributed and hence $R_0 \sim \text{Unif}[0,1]$. The main idea of the proof is that proposition \ref{prop:TauCDFbounds} yields bounds on $F_{R_0 \mid X}$, which we then invert to obtain bounds on $F_{Y_0}$. We then substitute these bounds into equation \eqref{eq:cdfY0givenX} to obtain bounds on $F_{Y_0 \mid X}(\cdot \mid 1)$. Inverting those bounds yields the quantile bounds given in section \ref{sec:treatmentEffectBounds}. Since the bounds of proposition \ref{prop:TauCDFbounds} are not always uniquely invertible, we approximate them by invertible bound functions. Here we explain the main argument, but we omit the full details since these are similar to the proof of proposition 2 in \cite{MastenPoirier2017}.

\bigskip

\textbf{Under $\mathcal{T}$-independence} of $Y_0$ from $X$ with $\mathcal{T} = [Q_{Y_0}(a),Q_{Y_0}(b)]$, we have $\mathcal{T}$-independence of $R_0$ from $X$ with $\mathcal{T} = [a,b]$. To see this, let $\tau \in [a,b]$. Then
\begin{align*}
	F_{R_0 \mid X}(\tau \mid x)
		&= \Prob(R_0 \leq \tau \mid X=x) \\
		&= \Prob(F_{Y_0}(Y_0) \leq \tau  \mid X=x) \\
		&= \Prob(Y_0 \leq Q_{Y_0}(\tau) \mid X=x) \\
		&= F_{Y_0 \mid X}(Q_{Y_0}(\tau) \mid x) \\
		&= F_{Y_0}(Q_{Y_0}(\tau)) \\
		&= \tau.
\end{align*}
The second line follows by definition of the rank $R_0$. The fifth line follows by $\mathcal{T}$-independence of $Y_0$ from $X$ and since $\tau \in [a,b]$. The third and sixth lines follow by A\ref{assn:continuity}.\ref{A1_1}. Thus proposition \ref{prop:TauCDFbounds} yields sharp bounds on $F_{R_0 \mid X}$. Substituting these bounds into our argument above yields the bounds on $Q_{Y_0 \mid X}(\tau \mid 1)$ in section \ref{sec:treatmentEffectBounds}.

Sharpness of these bounds holds in the same sense as sharpness of the CQTE bounds in proposition 3 of \cite{MastenPoirier2017}. That is, the conditional quantile of $Y_0 \mid X=1$ should be a continuous and strictly increasing function, while $\underline{Q}_{Y_0 \mid X}^\mathcal{T}(\tau \mid 1)$ and $\overline{Q}_{Y_0 \mid X}^\mathcal{T}(\tau \mid 1)$ may have discontinuities  and flat regions. Nevertheless we show there exists a function that is arbitrarily close (pointwise in $\tau$) to these bounds that is continuous and strictly increasing. To see this, for $\eta \in[0,1]$ define
\begin{align*}
	\underline{F}^{\mathcal{T}}_{R_0 \mid X}(u \mid 0;\eta)
		&= (1-\eta) \cdot \underline{F}^{\mathcal{T}}_{R_0 \mid X}(u \mid 0) + \eta \cdot u \\
	\overline{F}^{\mathcal{T}}_{R_0 \mid X}(u \mid 0;\eta)
		&= (1-\eta) \cdot \overline{F}^{\mathcal{T}}_{R_0 \mid X}(u \mid 0) + \eta \cdot u.
\end{align*}
These cdfs satisfy $\mathcal{T}$-independence. For each $\eta  > 0$, they are continuous and strictly increasing. Finally, they converge uniformly to $\underline{F}^{\mathcal{T}}_{R_0 \mid X}(u \mid 0)$ and $\overline{F}^{\mathcal{T}}_{R_0 \mid X}(u \mid 0)$, respectively, as $\eta \searrow 0$. Therefore, we can substitute the cdf bounds $\underline{F}^{\mathcal{T}}_{R_0 \mid X}(u \mid 0;\eta)$ and $\overline{F}^{\mathcal{T}}_{R_0 \mid X}(u \mid 0;\eta)$ into equation \eqref{eq:eq12MP2018}, invert and then substitute that into equation \eqref{eq:cdfY0givenX} to obtain
\[
	\overline{F}^{\mathcal{T}}_{Y_0 \mid X}(y \mid 1;\eta)
	\equiv
	\frac{\underline{F}_{R_0 \mid X}^{^{\mathcal{T}}\, -1}(F_{Y \mid X}(y \mid 0) \mid 0;\eta) - p_0 F_{Y \mid X}(y \mid 0)}{p_1}
\]
and
\[
	\underline{F}^{\mathcal{T}}_{Y_0 \mid X}(y \mid 1;\eta)
	\equiv
	\frac{\overline{F}_{R_0 \mid X}^{^{\mathcal{T}}\, -1}(F_{Y \mid X}(y \mid 0) \mid 0;\eta) - p_0 F_{Y \mid X}(y \mid 0)}{p_1}.
\]
Taking the inverses of these two cdfs and letting $\eta \searrow 0$ allows us to attain points arbitrarily close to the endpoints of the set $[\underline{Q}_{Y_0 \mid X}^\mathcal{T}(\tau \mid 1), \overline{Q}_{Y_0 \mid X}^\mathcal{T}(\tau \mid 1)]$. The rest of the interior is attained by selecting sufficiently small $\eta > 0$ and taking convex combinations of the bound functions, as in equation \eqref{eq:joint epsilon cdf}, and letting $\varepsilon$ vary from 0 to 1.

\bigskip

\textbf{For $\mathcal{U}$-independence} we also obtain sharpness of the interior because the functions $\underline{Q}_{Y_0 \mid X}^\mathcal{U}(\tau \mid 1)$ and $\overline{Q}_{Y_0 \mid X}^\mathcal{U}(\tau \mid 1)$ are not necessarily continuous or strictly increasing.  Nevertheless, as for $\mathcal{T}$-independence, we can obtain continuous and strictly increasing functions that are arbitrarily close (pointwise in $\tau$) to these bounds. To see this, for $\eta = (\eta_1,\eta_2)\in(0,\min\{p_1,p_0\})^2$ define
\begin{align*}
	\underline{p}_x(u;\eta)
		&= \min\{\max\{\underline{p}_x(u),\eta_1\},1-\eta_2\} \\
	\overline{p}_x(u;\eta)
		&= \min\{\max\{\overline{p}_x(u),\eta_1\},1-\eta_2\}
\end{align*}
where $\underline{p}_x$ and $\overline{p}_x$ are defined as in the proof of proposition \ref{prop:TauCDFbounds} and we let $U \equiv R_0$. 
These conditional probabilities lie in $(0,1)$ and satisfy $\mathcal{U}$-independence. Moreover, there exists an $\widetilde{\eta}_1 \in (0,\min \{p_1,p_0 \})$ such that for any $\eta_1 \in (0,\widetilde{\eta}_1)$, there is an $\eta _2(\eta_1) \in (0,\min\{p_1,p_0\})$ such that
\[
	\int_{0}^1 \underline{p}_x(u;\eta) \; du = p_x
\]
for $\eta = (\eta_1,\eta_2(\eta_1))$. This follows by the intermediate value theorem. An analogous result holds for the conditional probability $\overline{p}_x(u;\eta)$. For such values of $\eta$, define 
\[
	\underline{F}^{\mathcal{U}}_{R_0 \mid X}(u \mid 0;\eta)
	= \int_{0}^u \frac{\underline{p}_x(v;\eta)}{p_x} \; dv
	\qquad \text{and} \qquad
	\overline{F}^{\mathcal{U}}_{R_0 \mid X}(u \mid 0;\eta)
	= \int_{0}^u \frac{\underline{p}_x(v;\eta)}{p_x} \; dv.
\]
These cdf bounds are strictly increasing and continuous since the integrand is strictly positive. Therefore, we can substitute these cdf bounds into equation \eqref{eq:eq12MP2018}, invert and then substitute that into equation \eqref{eq:cdfY0givenX} to obtain
\[
	\overline{F}^{\mathcal{U}}_{Y_0 \mid X}(y \mid 1;\eta)
		\equiv \frac{
			\underline{F}_{R_0 \mid X}^{^{\mathcal{U}}\,-1}(F_{Y \mid X}(y \mid 0) \mid 0;\eta) 
				- p_0 F_{Y \mid X}(y \mid 0)
			}{
			p_1
			}
\]
and
\[
	\underline{F}^{\mathcal{U}}_{Y_0 \mid X}(y \mid 1;\eta)
		\equiv \frac{
			\overline{F}_{R_0 \mid X}^{^{\mathcal{U}}\, -1}(F_{Y \mid X}(y \mid 0) \mid 0;\eta) 
				- p_0 F_{Y \mid X}(y \mid 0)
			}{
			p_1
			}.
\]
Taking the inverses of these two cdfs and letting $\eta_1 \searrow 0$ allows us to attain points arbitrarily close to the endpoints of the set $[\underline{Q}_{Y_0 \mid X}^\mathcal{U}(\tau \mid 1), \overline{Q}_{Y_0 \mid X}^\mathcal{U}(\tau \mid 1)]$. The rest of the interior is attained by selecting sufficiently small $\eta_1 > 0$ and taking convex combinations of the bound functions, as in equation \eqref{eq:joint epsilon cdf}, and letting $\varepsilon$ vary from 0 to 1.
\end{proof}

\begin{proof}[Proof of corollary \ref{corr:ATT_Y0bounds}]
This result follows by derivations similar to the proof of corollary 1 in \cite{MastenPoirier2017}.
\end{proof}

\subsection{Proofs for Appendix \ref{sec:multivalTreat}}

\begin{proof}[Proof of theorem \ref{thm:AvgValueCharacterization_ctsX}]
Define $\widetilde{X} = \indicator(X > x)$. Now apply theorem \ref{thm:AvgValueCharacterization_discreteX} to $\widetilde{X}$.
\end{proof}

\begin{proof}[Proof of corollary \ref{corr:noRegDependence}]
Follows by defining $\widetilde{X} = \indicator(X > x)$ and applying corollary \ref{corr:monotonicPropensityScores}.
\end{proof}

\begin{proof}[Proof of theorem \ref{thm:AvgValueCharacterization_discreteX}]
By the law of iterated expectations, this is equivalent to showing that $\Prob(X=x \mid U \in [t_1,t_2]) = \Prob(X=x)$ for all $x \in \supp(X)$ and $t_1, t_2 \in\mathcal{T} \cup \{ 0, 1 \}$ with $t_1 < t_2$.

\begin{itemize}
\item[($\Rightarrow$)] Suppose $U$ is $\mathcal{T}$-independent of $X$. Let $t_1, t_2 \in\mathcal{T} \cup \{ 0, 1 \}$ with $t_1 < t_2$. Then, for any $x \in \supp(X)$,
\begin{align*}
	\Prob(X=x \mid U \in [t_1,t_2])
		&= \frac{\Prob(U \in [t_1,t_2] \mid X=x) \Prob(X=x)}{t_2-t_1} \\
		&= \frac{(\Prob(U \leq t_2 \mid X=x) - \Prob(U < t_1 \mid X=x)) \Prob(X=x)}{t_2-t_1} \\
    		&= \frac{(\Prob(U \leq t_2 \mid X=x) - \Prob(U \leq t_1 \mid X=x)) \Prob(X=x)}{t_2-t_1} \\
    		&= \Prob(X=x).
\end{align*}
The first equality follows from $U \sim \text{Unif}[0,1]$. The third equality follows since $U \mid X$ is continuously distributed, which itself follows by $X$ being discretely distributed and lemma \ref{lemma:continuity}. The fourth line follows from $\mathcal{T}$-independence.

\item[($\Leftarrow$)] Suppose that for any $x \in \supp(X)$,
\[
	\Prob(X=x \mid U \in [t_1,t_2]) = \Prob(X=x)
\]
for all $t_1, t_2 \in \mathcal{T} \cup \{ 0, 1 \}$ with $t_1 < t_2$. Then,
\begin{align*}
	\Prob(U \in [t_1,t_2] \mid X=x)
		&= \frac{\Prob(X=x \mid U \in[t_1,t_2]) \Prob(U \in [t_1,t_2])}{\Prob(X=x)} \\
		&= \frac{\Prob(X=x) \Prob(U \in [t_1,t_2])}{\Prob(X=x)} \\
    		&= \Prob(U \in [t_1,t_2]).
\end{align*}
The second line follows by assumption. Setting $t_1 = 0$ and using $U \sim \text{Unif}[0,1]$ gives the result.
\end{itemize}
\end{proof}

\begin{proof}[Proof of corollary \ref{corr:discreteXnotMonotone}]
Follows by defining $\widetilde{X} = \indicator(X\geq x)$ and applying corollary \ref{corr:monotonicPropensityScores}.
\end{proof}

\begin{proof}[Proof of corollary \ref{prop:TauImpliesU_ctsX}]
Follows by defining $\widetilde{X} = \indicator(X\leq x)$ and applying corollary \ref{prop:TauImpliesU}.
\end{proof}

\end{document}